\documentclass[12pt]{article}

\usepackage{amsfonts}
\usepackage{amsmath}
\usepackage{amsthm}
\usepackage{amssymb}
\usepackage[onehalfspacing]{setspace}
\usepackage[margin=1.25in, top=1.5in, bottom=1.5in]{geometry}
\usepackage{lscape}
\usepackage{graphicx}
\usepackage[hidelinks]{hyperref}
\usepackage{graphics}
\usepackage{multirow}
\usepackage{epsfig}
\usepackage{epstopdf}
\usepackage{comment}
\usepackage{float}

\allowdisplaybreaks

\usepackage{filecontents,catchfile,pgffor}

\usepackage{titling} 
\usepackage{blindtext}
\usepackage{booktabs,caption,fixltx2e}
\usepackage[flushleft]{threeparttable} 
\usepackage{hhline}  
\usepackage[utf8]{inputenc}
\usepackage[T1]{fontenc}
\usepackage[english]{babel}      
\usepackage{dsfont}   
\usepackage{comment} 
\usepackage{etoolbox}
\usepackage[all]{xy}
\usepackage{ulem}
\usepackage{threeparttable}
\usepackage{longtable}
\usepackage{multirow}
\usepackage{fancyhdr}
\usepackage{enumitem}
\usepackage{graphicx} 
\usepackage{siunitx}
\usepackage{tasks}
\usepackage{float}
\usepackage{amsmath}
\usepackage{subfiles}
\usepackage{listings}

\usepackage{verbatim}

\usepackage{array}
\usepackage{eurosym}

\usepackage{chngpage} 
\usepackage[final]{pdfpages}
\usepackage{placeins}
\usepackage{catchfile}

\usepackage{footnote}
\makesavenoteenv{tabular}
\makesavenoteenv{table}

\usepackage{titling}
\newcommand{\indep}{\perp \!\!\! \perp}
\DeclareMathOperator{\Var}{Var}
\DeclareMathOperator{\Cov}{Cov}

\usepackage{natbib}

\newtheorem{theorem}{Theorem}

\newtheorem{assumption}{Assumption}

\newtheorem{remark}{Remark}
\newtheorem{definition}{Definition}
\newtheorem{example}{Example}

\usepackage{enumitem}

\usepackage{scalerel}

\usepackage{chngcntr}

\usepackage[title]{appendix}

\usepackage{chngcntr}
\counterwithin*{equation}{section}

\usepackage{bibunits}

\begin{document}
\begin{bibunit}[aer]

\title{Predicting the Distribution of Treatment Effects:\\A Covariate-Adjustment Approach}

\author{Bruno Fava\thanks{Contact: bfava@stanford.edu. 
I am incredibly grateful to Federico Bugni and Eric Auerbach, Ivan Canay, Joel Horowitz and Dean Karlan for their exceptional advising and mentoring, and to Ahnaf Rafi for numerous helpful discussions. 
I thank participants of the econometrics seminar at Michigan State University and at the 2024 IPA \& GPRL Research Gathering, Midwest Econometrics Group Conference 2024, and Southern Economic Association 94th Annual Meeting for helpful comments. 
All errors are my own. 
This research was supported in part through the computational resources and staff contributions provided for the Quest high performance computing facility at Northwestern University which is jointly supported by the Office of the Provost, the Office for Research, and Northwestern University Information Technology.}\\Stanford University}
\date{October 8, 2026}

\maketitle

\vspace{-0.3in}

\begin{abstract}
    Important questions for impact evaluation require knowledge not only of average effects, but of the distribution of treatment effects. 
    The inability to observe individual counterfactuals makes answering these empirical questions challenging. 
    I propose an inference approach for points of the distribution of treatment effects that uses predicted counterfactuals through covariate adjustment. 
    I provide finite-sample valid inference using sample-splitting and asymptotically valid inference using cross-fitting under arguably weak conditions. 
    Revisiting five randomized controlled trials on microcredit that reported null average effects, I find important distributional impacts, with some individuals helped and others harmed by the increased credit access. 
\end{abstract}



\clearpage 

\section{Introduction} \label{section.introduction}

\sloppy
Important questions in impact evaluation cannot be answered by average treatment effects alone. 
What proportion of individuals are harmed by a treatment? 
Does a policy help many by a little? 
Or a few by a lot? 
The distribution of treatment effects carries crucial information with important equity and efficiency implications. %
For example, if a policy harms many despite a positive average effect, it may be crucial to reconsider its implementation. 
If benefits are concentrated among a few, more efficient targeting may be needed. 
However, learning the distribution of treatment effects is challenging because of the fundamental problem of causal inference: one can never observe individual counterfactuals. 
But one may be able to predict them.

I provide an approach to inference on points of the distribution of treatment effects, 
$$\theta = P (Y(1) - Y(0) \le \delta),$$ 
for any value of $\delta$, by leveraging counterfactual predictions via covariate adjustment. 
My setting consists of a representative sample of $(Y, D, X)$, where $D$ is a binary treatment, $Y = D Y(1) + (1 - D) Y(0)$ is the observed outcome, and $X$ is a vector of pre-treatment covariates. 
I assume either $D \indep (Y(1), Y(0), X)$, or unconfoundedness conditional on a known propensity score. 
This is typically the setting in Randomized Controlled Trials (RCTs). 
I refer to $\theta$ as the Distributional Treatment Effect (DTE) and to the method as CAIDE (Covariate Adjustment for Inference on Distributional Effects). 
All results are pointwise in $\delta$. 

My first contribution is a new characterization of the identified set of $\theta$ in the presence of covariates that has key advantages for inference.
While \citet{fan2010sharp} use the marginals of $Y(1)$ and $Y(0)$ for inference on $\theta$, I exploit the insight that, for any function of covariates $s$, it holds that
\begin{equation*} %
    \theta = P ( Y(1) - Y(0) \le \delta ) = P ( [Y(1) - s(X)] - [Y(0) - s(X)] \le \delta ),
\end{equation*}
so $\theta$ can be bounded using the marginals of covariate-adjusted outcomes $Y(1) - s(X)$ and $Y(0) - s(X)$. 
The choice of $s$ determines the width of the induced bounds, and informative functions of covariates can narrow the identified set relative to $s(x)=0$. 
Indeed, my main identification result characterizes optimal forms of $s$ that make the induced bounds sharp---equal to the smallest valid identified set given the distribution of the observable data $(Y, D, X)$---reducing the problem of bounding $\theta$ in the presence of covariates to one involving only two marginal distributions.

Building on this identification result, my main contribution is inference on $\theta$. 
The result naturally motivates a sample-splitting strategy: use one part of the data to estimate the optimal covariate-adjustment term, and the other to estimate the induced bounds and construct confidence sets. 
Because any covariate-adjustment term yields valid bounds, this construction delivers valid inference under weak conditions on how the term is estimated. 
In particular, I obtain a confidence interval that is finite-sample valid under no regularity conditions beyond the data coming from an RCT---to my knowledge, the first finite-sample inference result for the DTE that exploits covariates for obtaining a tighter identified set. 
This guarantee is especially valuable in high-stakes applications where exact finite-sample validity under minimal assumptions takes precedence over power, but can be conservative otherwise. 
I therefore propose a second procedure based on cross-fitting that uses the entire sample for both stages, yielding asymptotic validity at the expense of additional but arguably weak conditions. 
This procedure yields higher power than the finite-sample approach and previous literature under appropriate conditions.
Neither procedure restricts rates of convergence or model complexity, allowing the use of machine learning algorithms to estimate the covariate-adjustment term.

The relevance of my contribution is illustrated in an application to microcredit, where I find statistically significant evidence of positive and negative treatment effects from increased credit access. 
While academics and policymakers have hypothesized and debated the theory underlying such distributional impacts (\citealp{banerjee2013microcredit}, \citealp{garz2021consumer}), their empirical estimation has been hampered by methodological limitations. 
Meta-analyses suggest that average impacts are small if not negligible\footnote{\citet{breza2021measuring} find negative effects of the shutdown of microfinance institutions in Andhra Pradesh, which suggests that microcredit might have large and positive equilibrium effects.} (\citealp{banerjee2015six}, \citealp{meager2019understanding}), and analyses of heterogeneous effects mostly suggest positive impacts for some and rarely negative impacts\footnote{\citet{crepon2015estimating}, for example, find negative quantile treatment effect (QTE) on profits at the 10th percentile. 
However, QTEs are subject to the well-known limitation that they only recover the distribution of treatment effects under the strong assumption of rank-preservation; without it, they are just differences between marginal quantiles of $Y(1)$ and $Y(0)$ \citep{firpo2007efficient}, which may not have a causal interpretation. This assumption rules out, for example, that two individuals with the same value of $Y(0)$ may have different values of $Y(1)$. 
The authors also question whether this finding reflects the heterogeneity of credit access, citing low take-up and possible measurement error in the outcome.
}
(\citealp{angelucci2015microcredit}, \citealp{augsburg2015impacts}, \citealp{banerjeeMiracleMicrofinanceEvidence2015}, \citealp{banerjee2019can}, \citealp{meager2022aggregating}). 
I revisit five studies that randomized microloans at the individual level and find evidence of heterogeneity, including negative treatment effects for some households. 
The estimated lower bound on the proportion of individuals who would be better off without the loan varies from $2.8\%$ to $26.0\%$, and the lower bound on the proportion who benefited ranges from $2.9\%$ to $26.1\%$.

A growing literature studies inference on $\theta$.\footnote{On identification of $\theta$, \citet{makarov1982estimates} and \citet{williamson1990probabilistic} characterize the sharp bounds on the distribution of the sum of two random variables given only their marginals (including $\theta$ as a special case), and \citet{heckman1997making} and \citet{manski1997monotone} restrict the joint distribution of potential outcomes to achieve informative bounds on $\theta$.} 
Without covariates, \citet{fan2010sharp} propose asymptotically valid confidence intervals, and \citet{ruizHighConfidenceInference2024} provide a finite-sample approach. 
\citet{fan2010sharp} also characterize the sharp identified set of $\theta$ with covariates, but inference based on their characterization is challenging, and they do not pursue it (see Remark \ref{remark.fan_and_park}). 
\citet{kallus2022s} considers inference with covariates for binary outcomes under high-level assumptions on the data-generating process. 
Beyond the binary case, two recent working papers developed independently of this one, \citet{semenova2025debiased} and \citet{jiModelAgnosticCovariateAssistedInference2026} (henceforth SJLS), study the more general problem of averages of intersection bounds, which includes $\theta$ as a particular case. 
\citet{jiModelAgnosticCovariateAssistedInference2026} derive their estimator from the dual of an optimal transport problem and \citet{semenova2025debiased} from an envelope score, and their estimators share the same structure for the DTE. 

In contrast, by exploiting the specific structure of the DTE problem, my covariate-adjustment approach, which differs from SJLS in its starting point, enables a finite-sample valid confidence interval, and an asymptotic confidence interval that targets bounds always (weakly) narrower than those of SJLS by construction, yielding higher power under suitable conditions (Appendix~\ref{appendix.comparison}).
Moreover, I use the proportion of treated units in-sample instead of the true propensity score, leading to further power improvement (Online Appendix~\ref{appendix.prop_score}), and show asymptotic exactness under continuous outcomes.\footnote{\citet{jiModelAgnosticCovariateAssistedInference2026} derived exactness only under finite outcome support until their latest August 2026 version, which extends it to continuous outcomes under a margin condition.} 
In practice, these differences can reveal statistical significance that would otherwise be missed, as shown in a Monte Carlo study and the application to microcredit.

The paper proceeds as follows. 
Section \ref{section.identification} describes the setup and shows the main identification results of the paper. 
Section \ref{section.finite_sample} provides finite-sample inference using sample-splitting, and Section \ref{section.asymptotics} asymptotic inference and exactness using cross-fitting. 
Section \ref{section.simulation} presents a Monte Carlo study, and Section \ref{section.application} the application to microcredit. 
Finally, Section \ref{section.conclusion} concludes. 
Proofs not displayed in the main text are deferred to the Appendix.

\section{Setup and Identification via Covariate Adjustment} \label{section.identification}

In this section, I define the setup of the paper and establish new identification results for $\theta = P ( Y(1) - Y(0) \le \delta )$ that motivate the confidence intervals I introduce in Sections \ref{section.finite_sample} and \ref{section.asymptotics}. 
After presenting the setting in Assumption \ref{as.rct}, I characterize valid bounds on $\theta$ that are induced from using the marginal distributions of potential outcomes adjusted by a covariate-adjustment term $s(X)$ in~(\ref{def.induced_bounds}). 
Finally, I characterize optimal forms of $s$ that make such bounds sharp in Theorem \ref{th.best_s}, and discuss conditions for point identification in Section \ref{section.point_identification}. 

Without loss of generality, I focus on $\delta = 0$. For a different value of $\delta$, all results hold by redefining $\tilde{Y}(0) = Y(0) + \delta$.
Consider the standard potential outcomes setup, where $Y \in \mathcal{Y} \subseteq \mathbb{R}$ is the outcome, $D \in \{0,1\}$ is the treatment assignment indicator, $Y(1)$ and $Y(0)$ are potential outcomes connected by $Y = D Y(1) + (1 - D) Y(0)$, and $X \in \mathcal{X}$ denotes a set of pre-treatment covariates. 
The setting is formalized in Assumption \ref{as.rct}. 

\begin{assumption} \label{as.rct} Let $(Y, D, X) \sim P$. For some $Y_\ell, Y_u \in \mathbb{R}$ and $B > 0$, $P$ satisfies:
    \begin{enumerate}[label=\upshape(\roman*), ref= \ref{as.rct}(\roman*)]
        \item \label{as.bounded_y} \textit{(bounded outcome)} $Y \in [Y_\ell, Y_u]$;
        \item \label{as.unconfoundedness} \textit{(unconfoundedness)} $(Y(1), Y(0)) \indep D | X$; the propensity score $p(x) = P (D = 1 | X = x) = \pi$ does not depend on $x$ and satisfies $B < \pi < 1 - B$.
    \end{enumerate}
\end{assumption}

A\ref{as.bounded_y} of bounded $\mathcal{Y}$ is technical, and is discussed in Remark \ref{remark.finite_y}. 
A\ref{as.unconfoundedness} is used for identification of the CDFs of $Y(1)$ and $Y(0)$ conditional on covariates. 
It typically holds, for example, in RCTs, where $D$ is randomized. 
The constant propensity score assumption facilitates exposition, and I relax it for my asymptotic analysis, allowing for an arbitrary and known propensity score in Online Appendix \ref{appendix.relaxing_prop_score}. 
The case of an unknown propensity score that must be estimated is left for future work; in RCTs, $p(x)$ is known by design. 

The starting point for CAIDE is the fact that for any (measurable) function $s: \mathcal{X} \rightarrow \mathcal{Y}$, it holds that
\begin{equation} \label{eq.validity}
    \theta = P ( Y(1) - Y(0) \le 0 ) = P ( [Y(1) - s(X)] - [Y(0) - s(X)] \le 0 ).
\end{equation}
Therefore, I propose to bound $\theta$ using the marginal distributions of covariate-adjusted outcomes $Y(1) - s(X)$ and $Y(0) - s(X)$. 
To do so, I build on results of \citet{fan2010sharp}, who propose inference on $\theta$ using the marginals of $Y(1)$ and $Y(0)$. 
The main reason to incorporate covariates into the problem is to obtain narrower bounds on $\theta$ \citep{firpo2019partial}.

Before defining the bounds induced by covariate-adjustment, it is useful to revisit the definition of Makarov bounds (\citealp{makarov1982estimates}; \citealp{williamson1990probabilistic}), which are the sharp bounds on $\theta$ without covariates: 
\begin{equation} \label{def.theta_makarov}
    \theta \in \left[ \max_t \{ F_1(t) - F_0(t) \}, 1 + \min_t \{ F_1(t) - F_0(t) \} \right],
\end{equation}
where $F_j(t) = P ( Y(j) \le t )$, $j \in \{0,1\}$. 
By replacing the random variables $Y(1)$ and $Y(0)$ by $Y(1) - s(X)$ and $Y(0) - s(X)$, the bounds \textit{induced} by $s$ are defined as:
\begin{equation} \label{def.induced_bounds}
    \theta \in \left[ \max_t \Delta(t, s), \; 1 + \min_t \Delta(t, s) \right],
\end{equation} 
where
\begin{equation} \label{def.Delta}
    \Delta(t,s) = P ( Y(1) - s(X) \le t ) - P ( Y(0) - s(X) \le t ).
\end{equation}
Because of the equivalence in (\ref{eq.validity}) and the validity of the Makarov bounds in (\ref{def.theta_makarov}), (\ref{def.induced_bounds}) contains $\theta$ for any $s$. 
Note that if covariates are not used, e.g., if $s(x)=0$, then (\ref{def.induced_bounds}) is equivalent to (\ref{def.theta_makarov}). 

Next, I show that for specific covariate-adjustment terms, the bounds in (\ref{def.induced_bounds}) are equivalent to the sharp bounds (with covariates):
\begin{equation} \label{def.theta_sharp}
    [\theta_L^*, \theta_U^*] = \left[ \mathbb{E} \left[ \max_t \{ F_1(t | X) - F_0(t | X) \} \right], 1 + \mathbb{E} \left[ \min_t \{ F_1(t | X) - F_0(t | X) \} \right] \right],
\end{equation}
where $F_j(t | x) = P ( Y(j) \le t | X = x )$, $j=0,1$. 
This characterization and its proof of sharpness are due to \citet{fan2010sharp}. 
The bounds (\ref{def.theta_sharp}) are always weakly narrower than (\ref{def.theta_makarov}), as they exploit covariate information.
\citet{fan2010sharp} do not pursue inference based on (\ref{def.theta_sharp}), as it requires nonparametric estimation of $F_1( \cdot | X)$ and $F_0( \cdot | X)$ integrated over the distribution of $X$, rendering the asymptotic distribution intractable; see Remark \ref{remark.fan_and_park}. 
In contrast, inference based on Theorem \ref{th.best_s} below is tractable, and developed in the subsequent sections. 

\begin{theorem}[{Sharp transformations $s^*_L$ and $s^*_U$}] \label{th.best_s}
    Let Assumption \ref{as.unconfoundedness} hold. Define
    \begin{align} 
        s^*_L(x) & \in \arg \max_t \left\{ F_1(t | x) - F_0(t | x) \right\}, \label{def.best_sL} \\
        s^*_U(x) & \in \arg \min_t \left\{ F_1(t | x) - F_0(t | x) \right\}. \label{def.best_sU} 
    \end{align}
    Then, the bounds $\left[ \max_t \Delta(t, s^*_L), 1 + \min_t \Delta(t, s^*_U) \right] = [\theta_L^*, \theta_U^*]$ are sharp, the smallest valid bounds on $\theta$ given the distribution of the observable data $(Y, D, X)$. 
    That is, for every $\theta$ in this interval, there exists a joint distribution of $(Y(1), Y(0), D, X)$ with the same distribution of the observable data such that $P ( Y(1) - Y(0) \le 0 ) = \theta$. 
\end{theorem}

\begin{proof}
    For the lower bound,
    \begin{align}
        \theta_L^* & \ge \max_t \Delta(t, s^*_L) \label{proof.thsstar.1} \\ 
        & \ge \Delta(0, s^*_L) \nonumber \\
        & = P ( Y(1) - s^*_L(X) \le 0 ) - P ( Y(0) - s^*_L(X) \le 0 ) \label{proof.thsstar.3} \\
        & = \mathbb{E} \left[ P ( Y(1) - s^*_L(X) \le 0 | X ) - P ( Y(0) - s^*_L(X) \le 0 | X ) \right] \label{proof.thsstar.4} \\
        & = \mathbb{E} \left[ F_1(s^*_L(X) | X) - F_0(s^*_L(X) | X ) \right] \label{proof.thsstar.5} \\
        & = \mathbb{E} \left[ \max_t \left\{ F_1(t | X) - F_0(t | X) \right\} \right] \label{proof.thsstar.6} \\
        & = \theta_L^*, \label{proof.thsstar.7}
    \end{align}
    where (\ref{proof.thsstar.1}) holds from (\ref{def.induced_bounds}) and since $\theta_L^*$ is sharp \citep{fan2010sharp}; (\ref{proof.thsstar.3}) by definition of $\Delta$; (\ref{proof.thsstar.4}) by the law of iterated expectations; (\ref{proof.thsstar.5}) by definition of $F_j(t | x)$ ($j=0,1$); (\ref{proof.thsstar.6}) by the definition of $s^*_L$; and (\ref{proof.thsstar.7}) by definition of $\theta_L^*$ (\ref{def.theta_sharp}). 
    The upper bound case is analogous.
\end{proof}

Theorem \ref{th.best_s} is new: it shows that the sharp bounds on $\theta$ in the presence of covariates can be obtained via covariate adjustment. 
That is, for specific functions $s^*_L(x)$ and $s^*_U(x)$, the induced bounds in (\ref{def.induced_bounds}) are equal to the sharp bounds in (\ref{def.theta_sharp}). 
Note that $s^*_L(x)$ and $s^*_U(x)$ need not be unique and are generally different. 
Hence, in general, there is no single function $s$ that induces both lower and upper bounds to be sharp in (\ref{def.induced_bounds}), and in the next sections I propose to estimate both $s^*_L(x)$ and $s^*_U(x)$. 

\begin{remark} \label{remark.fan_and_park}
    (Comparison between Theorem \ref{th.best_s} and (\ref{def.theta_sharp})).

    My representation of the identified set in Theorem \ref{th.best_s} has two major advantages for inference compared to \citet{fan2010sharp}'s characterization in (\ref{def.theta_sharp}): inference is tractable when the conditional distributions are estimated nonparametrically, and robust to misspecification if they are estimated parametrically. 
    I show in Online Appendix \ref{appendix.fanpark} that the asymptotic distribution of plug-in estimators based on (\ref{def.theta_sharp}) depends on the specific estimation method used for the conditional distributions, even with sample-splitting. 
    It is difficult to characterize this asymptotic distribution without strong assumptions on the convergence rates of the conditional distribution estimators. 
    Since the bounds in (\ref{def.induced_bounds}) are valid for any $s$, inference based on Theorem \ref{th.best_s} is robust to misspecification: if $(s^*_L, s^*_U)$ are estimated inconsistently, for example from using a misspecified parametric estimator, the resulting bounds are wider but valid.  
\end{remark}

\begin{remark} \label{remark.finite_y}
    The assumption of bounded $\mathcal{Y}$ ensures that $s^*_L(x)$ and $s^*_U(x)$ are finite, which is important for the equivalence in (\ref{eq.validity}) when using the marginals of $Y(1) - s(X)$ and $Y(0) - s(X)$. 
    This assumption can be relaxed by defining a new outcome $\widetilde{Y} = \Phi(Y)$, where $\Phi$ is any bounded, strictly increasing function, such as the standard normal CDF. $\widetilde{Y}$ is always bounded, and $P ( Y(1) - Y(0) \le 0 ) = P ( \widetilde{Y}(1) - \widetilde{Y}(0) \le 0 ) = P ( [\widetilde{Y}(1) - s(X)] - [\widetilde{Y}(0) - s(X)] \le 0 )$ holds.
\end{remark}

\begin{remark}
    In principle, functions $s$ that differ from the optimal $s^*_L(x)$ and $s^*_U(x)$ can lead to wider bounds, potentially even wider than those obtained using $s(x) = 0$. 
    For example, if $X$ is independent of potential outcomes, using $s(x) = x$ with a single scalar covariate introduces noise, making the bounds wider. 
    This phenomenon, however, is not of concern when $s^*_L(x)$ and $s^*_U(x)$ are consistently estimated, since they yield the sharp bounds, which are by construction (weakly) narrower than those using $s(x)=0$. 
    Moreover, one can always calculate bounds with and without covariate adjustment and take their intersection. 
\end{remark}

\subsection{Conditions for Point Identification} \label{section.point_identification}

\begin{figure}[ht]
    \centering
    \begin{minipage}{\textwidth}
        \includegraphics[width=\linewidth]{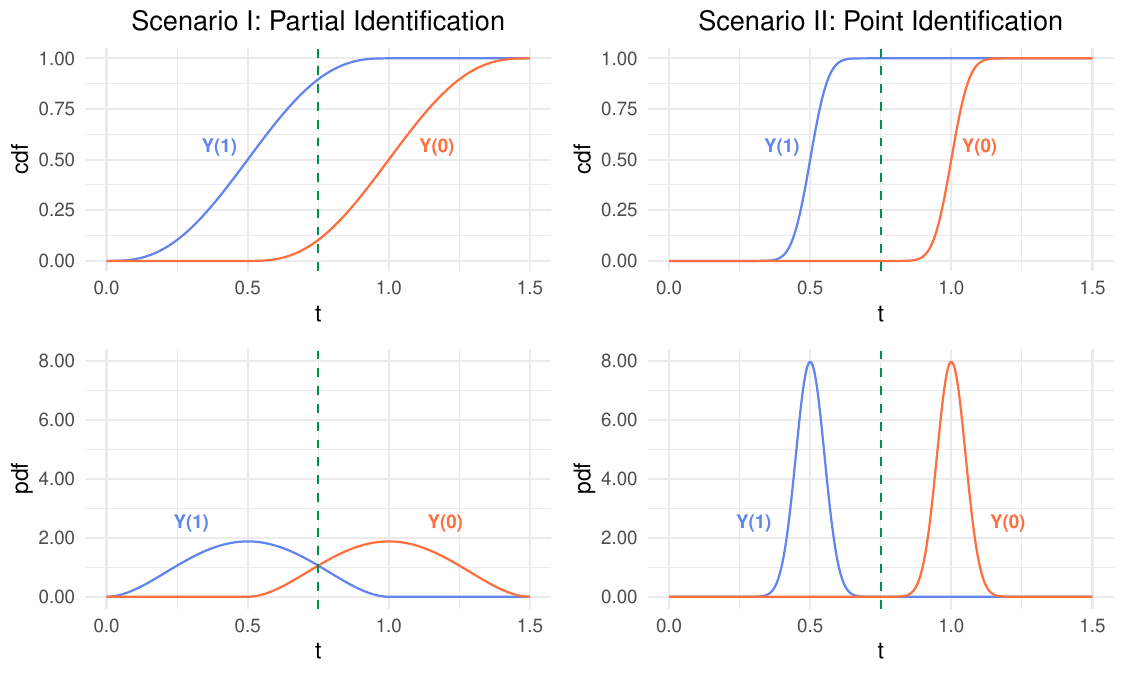}    
        {\footnotesize Blue lines represent the distribution of $Y(1)$ conditional on a given value of $X = x$, and red lines the conditional distribution of $Y(0)$. 
        Panels in the first row show the cumulative distributions, and in the second row, the probability density functions. 
        The left column shows a scenario with partial identification, and the right column shows a scenario with point identification. 
        The green dashed line represents the point $t$ that maximizes $F_1(t | x) - F_0(t | x)$. 
        In both scenarios $\min_t F_1(t | x) - F_0(t | x) = 0$, for example when $t = 0$.}
    \end{minipage}
    \caption{Example scenarios of distributions of potential outcomes conditional on a given value of covariates $X = x$.}
    \label{fig.identification}
\end{figure}

From (\ref{def.theta_sharp}), a sufficient condition for point identification of $\theta$ is that $\max_t \{ F_1(t | x) - F_0(t | x) \} = 1 + \min_t \{ F_1(t | x) - F_0(t | x) \}$ for all $x$. 
Figure \ref{fig.identification} gives intuition for when this condition holds, focusing on a single value of $X = x$ for ease of exposition. 
All panels show the distributions of $Y(1)$ (in blue) and $Y(0)$ (in red) conditional on $X = x$, which represent the problem of predicting the potential outcomes for this value of $x$. 
The upper panels show the cumulative distributions and the lower panels show the probability densities. 
Note that predictions are not a single point because potential outcomes are random variables even conditional on covariates. 
In both columns, $\min_t \{ F_1(t | x) - F_0(t | x) \} = 0$ (for example when $t = 0$), and $F_1(t | x) - F_0(t | x)$ is maximized at $t = 0.75$ (green dashed line). 

In the lower left panel, the distributions of $Y(1)$ and $Y(0)$ are partially overlapping, and the upper left panel reveals that $\max_t \{ F_1(t | x) - F_0(t | x) \} < 1 + \min_t \{ F_1(t | x) - F_0(t | x) \} = 1$. 
Since there is overlap, one cannot accurately predict if $Y(1) - Y(0) < 0$ or not. 
For example, if an individual is at the right tail of $Y(1)$ (blue) but at the left tail of $Y(0)$ (red), then $Y(1) - Y(0) > 0$. 
If the opposite happens, then $Y(1) - Y(0) < 0$. 
However, since the overlap is small, the probability that $Y(1) - Y(0) < 0$ is large, and the fraction harmed (for this value of $X=x$) must be close to one. 
In the lower right panel, the distributions are disjoint, and $\max_t \{ F_1(t | x) - F_0(t | x) \} = 1 + \min_t \{ F_1(t | x) - F_0(t | x) \} = 1$ (upper right panel). 
In this case, even though the covariates are not informative enough to get a point prediction of $Y(1)$ and $Y(0)$, they are informative enough to learn that $Y(1) < Y(0)$ with probability one for this value of $X = x$, and thus conclude that all individuals with $X = x$ are harmed. 

Note that the interval in (\ref{def.theta_sharp}) is an average of $\max_t \{ F_1(t | X) - F_0(t | X) \}$ and $1 + \min_t \{ F_1(t | X) - F_0(t | X) \}$ over values of $X$. 
Hence, the figure illustrates that the informativeness of covariates in predicting the outcomes is crucial for having a narrow interval for $\theta$. 
In fact, if the predictions (conditional distributions) of $Y(1)$ and $Y(0)$ are disjoint for all values of $X$, then $\theta$ is point-identified. 
Note that in Figure \ref{fig.identification}, the green line represents the value of $s^*_L(x)$, which is the point $t$ that maximizes the distance $F_1(t | x) - F_0(t | x)$. 
Hence, $s^*_L(x)$ is the point that makes the two distributions as \textit{separated} as possible. 
This fact is what motivates the notation $s$, for \textit{separator} function. 

\section{Finite-Sample Inference with Sample-Splitting} \label{section.finite_sample}

In this section, I propose sample-splitting estimators for $(\theta_L^*, \theta_U^*)$ and a confidence interval (CI) for $\theta$ that is valid in finite samples.
This guarantee is especially useful in high-stakes applications where one is willing to trade off power to have a CI that is always valid, as opposed to approximately valid in large samples.
Given an i.i.d. sample $(Y_i, D_i, X_i)_{i=1}^n \sim P$ that satisfies A\ref{as.rct}, the estimation strategy consists of randomly splitting the sample into two, using the first set to estimate $(s^*_L, s^*_U)$ with $(\hat{s}_L, \hat{s}_U)$, and the second set to estimate the bounds $[\max_t \Delta(t, \hat{s}_L), 1 + \min_t \Delta(t, \hat{s}_U)]$ (as in (\ref{def.induced_bounds})) with a sample analogue. 
Conditional on the first sample, these bounds are valid due to (\ref{eq.validity}), and finite-sample inference is achieved by applying the DKW inequality \citep{dvoretzky1956asymptotic, massart1990tight}.

I propose sample-splitting estimators for the sharp bounds $(\theta_L^*, \theta_U^*)$:
\begin{definition}[Sample-splitting estimators] \label{def.estimator_ss} \,
    \begin{itemize}
        \item[(i)] Randomly split the sample into two sets, denoted \textit{main} and \textit{auxiliary} sets. Denote by $\mathcal{I}_M^{1} \subset \{1, \dots, n \}$ the units in the main sample assigned to treatment, and $\mathcal{I}_M^{0}$ the units assigned to control;
        \item[(ii)] Using data from the auxiliary sample only, estimate $(s^*_L, s^*_U)$ by replacing $F_1(t | X)$ and $F_0(t | X)$ in (\ref{def.best_sL}) and (\ref{def.best_sU}) by any estimators of the conditional CDFs $\hat{F}_1(t | X)$ and $\hat{F}_0(t | X)$. Denote the estimators $(\hat{s}_L, \hat{s}_U)$.
        \item[(iii)] In the main sample, compute the final estimators
        \begin{align*}
        \widetilde{\theta}_L & = \max_{t} \left\{ \frac{1}{|\mathcal{I}^{1}_M|} \sum_{i \in \mathcal{I}^{1}_M} \mathbb{I} \left(Y_i - \hat{s}_{L}(X_i) \le t \right) - \frac{1}{|\mathcal{I}^{0}_M|} \sum_{i \in \mathcal{I}^{0}_M} \mathbb{I} \left(Y_i - \hat{s}_{L}(X_i) \le t \right) \right\}, \\% \label{theta_hat_L} \\
        \widetilde{\theta}_U & = 1 + \min_{t} \left\{ \frac{1}{|\mathcal{I}^{1}_M|} \sum_{i \in \mathcal{I}^{1}_M} \mathbb{I} \left(Y_i - \hat{s}_{U}(X_i) \le t \right) - \frac{1}{|\mathcal{I}^{0}_M|} \sum_{i \in \mathcal{I}^{0}_M} \mathbb{I} \left(Y_i - \hat{s}_{U}(X_i) \le t \right) \right\}. %
        \end{align*}
    \end{itemize}
\end{definition}
$F_1(t | X)$ and $F_0(t | X)$ can be estimated with a variety of methods, including machine learning algorithms or other nonparametric estimators (see, e.g., \citealp{kneib2023rage} for a review). 
I discuss guidelines for estimating $\hat{s}$ in Online Appendix \ref{appendix.ml}. 
Although \textit{any} method can be used to estimate $F_1(t | X)$ and $F_0(t | X)$, estimates $(\hat{s}_L, \hat{s}_U)$ closer to $(s^*_L, s^*_U)$ yield bounds $[\max_t \Delta(t, \hat{s}_L), 1 + \min_t \Delta(t, \hat{s}_U)]$ that are closer to the sharp $[\max_t \Delta(t, s^*_L), 1 + \min_t \Delta(t, s^*_U)]$, and poor estimators can lead to wider intervals.

A default choice is to make the main and auxiliary samples roughly the same size, but that is unnecessary. 
In general, a larger auxiliary sample helps estimate $(\hat{s}_L, \hat{s}_U)$ closer to $(s^*_L, s^*_U)$, which improves the length of the target bounds $[\max_t \Delta(t, \hat{s}_L), 1 + \min_t \Delta(t, \hat{s}_U)]$. 
On the other hand, a larger auxiliary sample implies a smaller main sample, which leads to larger variance of $(\widetilde{\theta}_L, \widetilde{\theta}_U)$ around the target bounds.

\begin{theorem}[\hyperlink{proof.th.fs_cis}{Finite-sample validity of sample-splitting CI}] \label{th.fs_cis}
    Let Assumption \ref{as.rct} hold. Then, for any $\alpha \in (0,1)$ and any estimators $\hat{s}_{L}$ and $\hat{s}_{U}$, it holds that
    $$P \left( \theta \in \left[\widetilde{\theta}_L - c_\alpha, 1 \right] \right) \ge 1 - \alpha, \; \; P \left( \theta \in \left[ 0, \widetilde{\theta}_U + c_\alpha \right] \right) \ge 1 - \alpha,$$
    $$P \left( \theta \in \left[ \widetilde{\theta}_L - c_{\alpha/2}, \widetilde{\theta}_U + c_{\alpha/2} \right] \right) \ge 1 - \alpha,$$
    where
    $$c_\alpha = \left( \frac{\log(2 / \alpha)}{2} \right)^\frac{1}{2} \left( {|\mathcal{I}^{1}_M|}^{-\frac{1}{2}} + {|\mathcal{I}^{0}_M|}^{-\frac{1}{2}} \right).$$
\end{theorem}

Theorem \ref{th.fs_cis} provides confidence intervals that cover $\theta$ with probability at least $1 - \alpha$ for any size of the main and auxiliary samples. 
Since $\theta \in [0,1]$, the one-sided CIs are particularly useful for testing $\theta=0$ or $\theta=1$. 
The proof of the theorem exploits sample splitting to take $\hat{s}_L, \hat{s}_U$ as fixed in the main sample. 
Then, since by (\ref{def.induced_bounds}) $\theta$ must lie within $[\max_t \Delta(t, \hat{s}_L), 1 + \min_t \Delta(t, \hat{s}_U)]$, the CIs are constructed to cover this interval.

To my knowledge, this is the first finite-sample inference result that exploits covariates for narrowing the identified set of $\theta$. 
The result is a consequence of the new covariate-adjustment construction: since any $\hat{s}_L, \hat{s}_U$ yield valid bounds by (\ref{def.induced_bounds}), the guarantee requires no condition on how they are estimated.

The finite-sample guarantee comes from the DKW inequality \citep{dvoretzky1956asymptotic, massart1990tight}, a concentration inequality for empirical CDFs. 
A similar approach was implemented in \citet{ruizHighConfidenceInference2024} for finite-sample inference on $\theta$ without covariates. 
My proof differs from \citet{ruizHighConfidenceInference2024}'s in that I account for sample splitting, and I use the one-sided DKW inequality of \citet{massart1990tight} to obtain a smaller critical value for the one-sided CIs. 

Note that a large main sample leads to $c_\alpha$ being close to zero, and a large auxiliary sample makes a suitable choice of $(\hat{s}_L, \hat{s}_U)$ closer to $(s^*_L, s^*_U)$. Therefore, under regularity conditions, the CIs in Theorem \ref{th.fs_cis} collapse into the sharp bounds in (\ref{def.theta_sharp}) as both sample sizes grow to infinity. 
The next section investigates the large-sample behavior of a similar estimator, which allows for cross-fitting instead of sample-splitting.

\section{Large-Sample Inference with Cross-Fitting} \label{section.asymptotics}

In this section, I propose confidence intervals for $\theta$ based on cross-fitting that improve power over the sample-splitting approach of Section \ref{section.finite_sample}. 
While that approach provides a strong guarantee of coverage under weak conditions, it may have low power for two reasons: it uses only part of the data to estimate $(\hat{s}_L, \hat{s}_U)$ and the induced bounds, and the DKW-based critical values may be conservative for many distributions $P$. 
The cross-fitting approach addresses the first issue by using the full sample for both tasks, and the second by replacing the DKW critical values with asymptotic ones that are exact for any $P$ under suitable conditions.
These gains come at the cost of additional regularity conditions and a weaker coverage guarantee: asymptotic rather than exact finite-sample validity.
I define the conditions in A\ref{as.asymptotic} and show in Theorem \ref{th.asymptotic_dist_simplified} an intermediate result, that the estimators are asymptotically Gaussian. 
The main result is then established in Theorem \ref{th.cis}, which shows uniform validity of CIs for $\theta$. 
Finally, Theorem \ref{th.exactness} shows conditions under which the CIs are exact.

\begin{definition}[Cross-fitting estimators] \label{def.estimator_cf} \,
    \begin{enumerate}[label=\upshape(\roman*), ref= \ref{def.estimator_cf}(\roman*)]
        \item Randomly split the treatment sample into $K$ sets of the same size, that is, a partition of $\{1 \le i \le n : D_i = 1 \}$ into equal-sized folds $\{ \mathcal{I}^1_1, \dots, \mathcal{I}^1_K \}$. Similarly, split the control sample into $\{ \mathcal{I}^0_1, \dots, \mathcal{I}^0_K \}$, and let $\mathcal{I}_k = \mathcal{I}^1_k \cup \mathcal{I}^0_k$;%
        \item For all $k \in \{1, \dots, K \}$, use data from all folds except $\mathcal{I}_k$ to fit estimators $\hat{s}_{L}, \hat{s}_{U}$ to target $s^*_L, s^*_U$, as in Definition \ref{def.estimator_ss}(ii). Denote by $\hat{s}_{L, k}, \hat{s}_{U, k}$ the estimated functions for each $k$;
        \item \label{def.theta_hat_cf} Denote by $k(i)$ the index of the fold such that $i \in \mathcal{I}_{k(i)}$, $n_1 = \sum_{i=1}^n D_i$ and $n_0 = \sum_{i=1}^n 1 - D_i$. Compute the final estimators
        \begin{align*}
            \hat{\theta}_{L} & = \max_t \Biggl\{ \frac{1}{n_1} \sum_{i : D_i = 1} \mathbb{I}(Y_i \le \hat{s}_{L, k(i)}(X_i) + t) - \frac{1}{n_0} \sum_{i : D_i = 0} \mathbb{I}(Y_i \le \hat{s}_{L, k(i)}(X_i) + t) \Biggr\}, \\
            \hat{\theta}_{U} & = 1 + \min_t \Biggl\{ \frac{1}{n_1} \sum_{i : D_i = 1} \mathbb{I}(Y_i \le \hat{s}_{U, k(i)}(X_i) + t) - \frac{1}{n_0} \sum_{i : D_i = 0} \mathbb{I}(Y_i \le \hat{s}_{U, k(i)}(X_i) + t) \Biggr\}.
        \end{align*}
    \end{enumerate}
\end{definition}
Again, several methods can be used to estimate $(\hat{s}_{L,k}, \hat{s}_{U,k})$, including machine learning algorithms or other nonparametric approaches (see, e.g., \citealp{kneib2023rage} for a review). 
$K$ is assumed fixed as $n \to \infty$, and typical choices are $K=5$ or $K=10$.

To derive the asymptotic properties of the cross-fitting estimators, I introduce the following regularity conditions. 
The results are uniform over a set of probability measures $\mathcal{P}$, as the literature on inference in partially identified models argues that uniformity is necessary to obtain approximations that accurately represent finite-sample behavior (e.g., \citealp{imbens2004confidence, andrews2010inference}). 
To this end, I add a subscript $P$ throughout: $\Delta_P$ for (\ref{def.Delta}); $s^*_{L,P}$ and $s^*_{U,P}$ for the sharp transformations in (\ref{def.best_sL})--(\ref{def.best_sU}); and $(\theta_{L,P}^*, \theta_{U,P}^*)$ for the sharp bounds in (\ref{def.theta_sharp}).

\begin{assumption} \label{as.asymptotic} The set of probability measures $\mathcal{P}$ satisfies:
    \begin{enumerate}[label=\upshape(\roman*), ref= \ref{as.asymptotic}(\roman*)]%
        \item \label{as.limit_of_s} Limit of $\hat{s}_{L}$ and $\hat{s}_{U}$: There exist $s_{L,P}, s_{U,P}$ for all $P \in \mathcal{P}$ such that, for all $\varepsilon > 0$, $\sup_{P \in \mathcal{P}} P(|\hat{s}_{L}(X) - s_{L,P}(X)| > \varepsilon) \to 0$ and $\sup_{P \in \mathcal{P}} P(|\hat{s}_{U}(X) - s_{U,P}(X)| > \varepsilon) \to 0$ as $n \to \infty$;
        \item \label{as.cont_outcomes} Continuity of distributions of adjusted potential outcomes: for $A \in \{L, U\}$ and $j=0,1$, the set of CDFs $\{ t \mapsto P( Y(j) - s_{A,P}(X) \le t ) : P \in \mathcal{P} \}$ is uniformly equicontinuous;
        \item \label{as.uniqueness_t_infty} Uniqueness of optimizers: For $A \in \{L, U\}$, let 
        $$d_{A,P}(t, u) = \sum_{j=0}^1 | P(Y(j) \le s_{A,P}(X) + t) - P(Y(j) \le s_{A,P}(X) + u) |.$$
        There exist $t_{max, P}, t_{min, P}$ for all $P \in \mathcal{P}$ such that, for all $\varepsilon > 0$,
        \vspace{-\baselineskip}
        \begin{align*}
            \sup_{P \in \mathcal{P}} \sup_{t: d_{L,P}(t, t_{max,P}) \ge \varepsilon} & \Delta_P(t, s_{L,P}) - \Delta_P(t_{max,P}, s_{L,P}) < 0, \\ 
            \inf_{P \in \mathcal{P}} \inf_{t: d_{U,P}(t, t_{min,P}) \ge \varepsilon} & \Delta_P(t, s_{U,P}) - \Delta_P(t_{min,P}, s_{U,P}) > 0;
        \end{align*}
        \item \label{as.bounded_variance} Lower bounded variance: For some $C>0$ and for $j = 0, 1$, $C < P(Y(j) \le s_{L,P}(X) + t_{max,P}) < 1 - C$ and $C < P(Y(j) \le s_{U,P}(X) + t_{min,P}) < 1 - C$ for all $P \in \mathcal{P}$.
    \end{enumerate}
\end{assumption}

A\ref{as.limit_of_s} requires the estimators $(\hat{s}_{L}, \hat{s}_{U})$ to have a limit in probability, at any rate of convergence, and allows them to be inconsistent for the sharp $(s^*_{L,P}, s^*_{U,P})$. %
A\ref{as.cont_outcomes} restricts the possibility of point masses in the distribution of the adjusted outcomes $Y(j) - s_{A,P}(X)$. 
For example, it rules out limits $s_{A,P}$ that predict $Y(j)$ exactly with positive probability. 
A\ref{as.uniqueness_t_infty} is assumed to simplify exposition. 
It ensures the asymptotic distribution is first-order insensitive to the choice of the optimizers in Definition \ref{def.theta_hat_cf}, and it is a generalization of the assumption used in \citet{fan2010sharp} in their approach to inference on Makarov bounds without covariates. 
It requires the optimizer to be unique only up to thresholds that do not change the indicators $\mathbb{I}(Y(j) \le s_{A,P}(X) + t)$: for $t \le u$, $d_{A,P}(t, u) = 0$ exactly if $P(t < Y(j) - s_{A,P}(X) \le u) = 0$ for $j=0,1$.
Online Appendix \ref{appendix.relaxing_maximizer} discusses inference without A\ref{as.uniqueness_t_infty}, either by modifying the estimators to avoid the optimization over $t$, or by keeping them and computing critical values following \citet{chernozhukovIntersectionBoundsEstimation2013} or \citet{fangInferenceDirectionallyDifferentiable2019}. 
Example \ref{ex.unique_optimizers} illustrates a data-generating process where A\ref{as.uniqueness_t_infty} holds. 
A\ref{as.bounded_variance} ensures that $\Var[\sqrt{n} \hat{\theta}_{L}]$ and $\Var[\sqrt{n} \hat{\theta}_{U}]$ do not converge to zero, and it is only used to apply the result from \citet{stoye2009more} to get asymptotic validity of the two-sided CI in Theorem \ref{th.cis}. 

\begin{example}[Uniqueness of optimizers] \label{ex.unique_optimizers}
    Consider a data-generating process where $Y(0) = f(X) + \varepsilon$, $Y(1) = Y(0) + \tau$, $\varepsilon, \tau \indep X$ and $\mathbb{E}[Y(1)] = \mathbb{E}[Y(0)] = 0$. 
    For simplicity, I focus on a single probability measure $P$, that is, $\mathcal{P}$ is a singleton. 
    When the distributions of $Y(1)$ and $Y(0)$ are continuous and symmetric around zero, the optimal covariate-adjustment term is the same for both the lower and upper bounds and is given by $s^*(x) = \mathbb{E}[Y(0) | X = x] = f(x)$. 
    The distributions of $Y(1) - s^*(X)$ (in blue) and $Y(0) - s^*(X)$ (in red) are shown in Figure \ref{fig.example_optimizer} below. 
    The difference $\Delta_P(t, s^*) = P(Y(1) - s^*(X) \le t) - P(Y(0) - s^*(X) \le t)$ is uniquely maximized and minimized respectively at the green and yellow vertical dashed lines, where the probability density functions intersect. 
    \begin{figure}[ht]
        \centering
        \begin{minipage}{\textwidth}
            \includegraphics[width=\linewidth]{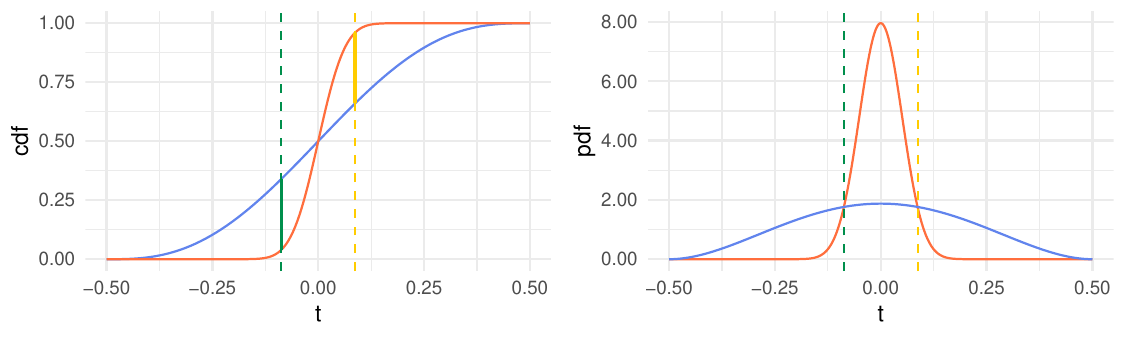}    
            {\footnotesize Visualization of Example \ref{ex.unique_optimizers}. Blue lines represent the distribution of $Y(1) - s^*(X)$, and red lines the distribution of $Y(0) - s^*(X)$. The left panel shows the cumulative distribution function and the right panel shows the probability density function. 
            The green dashed line represents the unique point $t$ where $\Delta_P(t, s^*)$ is maximized, and the solid green line in the left panel represents the magnitude of $\Delta_P(t, s^*)$ at this point. Similarly, the yellow dashed and solid lines show the unique point that minimizes $\Delta_P(t, s^*)$ and its magnitude.}
        \end{minipage}
        \caption{Example of data-generating process that satisfies A\ref{as.uniqueness_t_infty}.}
        \label{fig.example_optimizer}
    \end{figure}
\end{example}

Under A\ref{as.uniqueness_t_infty}, the presence of the optimizers in Definition \ref{def.theta_hat_cf} is asymptotically irrelevant, and $(\hat{\theta}_{L}, \hat{\theta}_{U})$ behave as sample averages. 
It follows that their asymptotic distribution is $\sqrt{n}$-Gaussian, as shown in Theorem \ref{th.asymptotic_dist_simplified}, an intermediate result for establishing confidence intervals for $\theta$. 
The presence of the estimated $(\hat{s}_{L, k}, \hat{s}_{U, k})$ will only affect where $(\hat{\theta}_{L}, \hat{\theta}_{U})$ are centered at, but these target bounds will always contain $\theta$, since any $s$ induces valid bounds (\ref{eq.validity}).

\begin{theorem}[\hyperlink{proof.th.asymptotic_dist}{Asymptotic distribution of cross-fitting estimators}] \label{th.asymptotic_dist_simplified}
    Let $\mathcal{P}$ denote a set of probability measures satisfying Assumptions \ref{as.rct} and \ref{as.asymptotic}. Then, there exist random variables $\bar{\theta}_{L, P}, \bar{\theta}_{U, P}$ with $\bar{\theta}_{L, P} \le \theta \le \bar{\theta}_{U, P}$ almost surely such that
    $$\sqrt{n} \begin{bmatrix} \hat{\theta}_{L} - \bar{\theta}_{L, P} \\ \hat{\theta}_{U} - \bar{\theta}_{U, P} \end{bmatrix} \overset{d}{\to} \mathcal{N} \left( \begin{bmatrix} 0 \\ 0 \end{bmatrix}, \begin{bmatrix} \sigma^2_{L,P} & \sigma_{L,U,P} \\ \sigma_{L,U,P} & \sigma^2_{U,P} \end{bmatrix} \right)$$
    uniformly in $P \in \mathcal{P}$.
\end{theorem}

Definitions of $\sigma^2_{L,P}$, $\sigma^2_{U,P}$ and $\sigma_{L,U,P}$ are provided in Appendix \ref{appendix.definitions}.  
The cross-fitting estimators are asymptotically normal at the usual $\sqrt{n}$ rate, centered at bounds $[\bar{\theta}_{L, P}, \bar{\theta}_{U, P}]$ that may be wider than the sharp $[\theta_L^*, \theta_U^*]$, but which always contain $\theta$. 
The nonparametric rates of convergence of $\hat{s}_L, \hat{s}_U$ only affect the target bounds $[\bar{\theta}_{L, P}, \bar{\theta}_{U, P} ]$. 
The resulting $\sqrt{n}$ rate is made possible through the splitting of the sample. 
Conditional on the estimated $(\hat{s}_{L, k}, \hat{s}_{U, k})$ and under A\ref{as.uniqueness_t_infty}, $(\hat{\theta}_{L}, \hat{\theta}_{U})$ behave as sample averages. 
Conditions under which $(\bar{\theta}_{L, P}, \bar{\theta}_{U, P}) = (\theta_{L,P}^*, \theta_{U,P}^*)$ are given in Theorem \ref{th.exactness}.

Confidence intervals for $\theta$ follow from consistent sample-analogue estimators for $\sigma_{L}^2$, $\sigma_{U}^2$ and $\sigma_{L,U}$. 
I provide expressions in Appendix \ref{appendix.definitions}. 
One-sided confidence intervals follow directly from the normal approximation of Theorem \ref{th.asymptotic_dist_simplified}. 
Since $\theta \in [0,1]$, the one-sided CIs are particularly useful for testing $\theta=0$ or $\theta=1$. 
For two-sided inference on $\theta$, I suggest using the confidence interval $\text{CI}^3_\alpha$ proposed in \citet{stoye2009more}, here denoted by $\widehat{CI}_\alpha$. 
A precise definition is reproduced in Online Appendix \ref{appendix.stoye}. 

\begin{theorem}[\hyperlink{proof.th.cis}{Uniform validity of confidence intervals}] \label{th.cis}
    Let $\mathcal{P}$ denote a set of probability measures satisfying Assumptions \ref{as.rct} and \ref{as.asymptotic}, and denote by $\Theta(P) = [\theta_{L,P}^*, \theta_{U,P}^*]$ the sharp identified set for $\theta$ when $(Y, D, X) \sim P \in \mathcal{P}$. Then,
    \begin{align*}
        & \liminf_{n \to \infty} \inf_{P \in \mathcal{P}} \inf_{\theta \in \Theta(P)} P \left( \theta \ge \hat{\theta}_{L} - z_\alpha \frac{\hat{\sigma}_{L}}{\sqrt{n}} \right) \ge 1 - \alpha, \\
        & \liminf_{n \to \infty} \inf_{P \in \mathcal{P}} \inf_{\theta \in \Theta(P)} P \left( \theta \le \hat{\theta}_{U} + z_\alpha \frac{\hat{\sigma}_{U}}{\sqrt{n}} \right) \ge 1 - \alpha, \\
        & \liminf_{n \to \infty} \inf_{P \in \mathcal{P}} \inf_{\theta \in \Theta(P)} P \left( \theta \in \widehat{CI}_\alpha \right) \ge 1 - \alpha,
    \end{align*}
where $z_\alpha = \Phi^{-1}(1 - \alpha)$, the $(1-\alpha)$-th quantile of the standard normal distribution.%
\end{theorem}

Since the cross-fitting estimators depend on the random partition of the sample into $K$ folds, the confidence intervals in Theorem \ref{th.cis} may lead to different conclusions in terms of statistical significance across runs. 
To address this, the cross-fitting procedure can be repeated several times with independent random splits, and the results averaged across repetitions. 
\citet{favaTrainingTestingMultiple2025} shows that asymptotic normality is preserved under such repeated splitting, and a valid standard error can be obtained by averaging the standard errors across repetitions.

Finally, I show that the target bounds $(\bar{\theta}_{L, P}, \bar{\theta}_{U, P})$ converge to the sharp $(\theta^*_{L, P}, \theta^*_{U, P})$ if the estimators of the covariate-adjustment terms are consistent for the optimal ones at any rate. 
Moreover, I show that the confidence intervals are asymptotically exact under a smoothness and rate condition. 
\begin{theorem}[\hyperlink{proof.th.exactness}{Asymptotic exactness of confidence intervals}] \label{th.exactness}
    Let $\mathcal{P}$ denote a set of probability measures satisfying Assumptions \ref{as.rct} and \ref{as.asymptotic}. Then, uniformly in $P \in \mathcal{P}$, $\bar{\theta}_{L, P} - \theta_{L, P}^* \overset{P}{\longrightarrow} 0$ and $\bar{\theta}_{U, P} - \theta_{U, P}^* \overset{P}{\longrightarrow} 0$ if respectively $s_{L,P}(x) = s^*_{L,P}(x)$ and $s_{U,P}(x) = s^*_{U,P}(x)$. Moreover, if
    \begin{itemize}
        \item[(i)] The CDFs $P( Y(j) \le t | X = x)$ for $j=0,1$ are twice continuously differentiable in $t$ with second derivative bounded uniformly in $P \in \mathcal{P}$ and $x \in \mathcal{X}$;
        \item[(ii)] Uniformly in $P \in \mathcal{P}$,
        $\sqrt{n} \mathbb{E}_{P} \left[ \left( \hat{s}_{L,k}(X) - s^*_{L, P}(X) \right)^2 | \hat{s}_{L,k} \right] \overset{P}{\longrightarrow} 0$ and 
        
        $\sqrt{n} \mathbb{E}_{P} \left[ \left( \hat{s}_{U,k}(X) - s^*_{U, P}(X) \right)^2 | \hat{s}_{U,k} \right] \overset{P}{\longrightarrow} 0;$
    \end{itemize}
        then, $\sqrt{n}(\bar{\theta}_{L, P} - \theta_{L,P}^*) \overset{P}{\longrightarrow} 0$ and $\sqrt{n}(\bar{\theta}_{U, P} - \theta_{U,P}^*) \overset{P}{\longrightarrow} 0$ uniformly in $P \in \mathcal{P}$. 
\end{theorem}

Note that a simple sufficient condition for (ii) in Theorem \ref{th.exactness} is that $|\hat{s}_{L,k}(X) - s^*_{L, P}(X)|^2$ converges to zero at the semiparametric rate $o_{P}(n^{-1/4})$. %

\section{Monte Carlo Simulations} \label{section.simulation}

I illustrate the results of Theorems \ref{th.fs_cis} and \ref{th.asymptotic_dist_simplified} with a simulation study. 
I also compare the performance of the sample-splitting and cross-fitting estimators, as well as the estimator proposed by \citet{semenova2025debiased} and \citet{jiModelAgnosticCovariateAssistedInference2026} (SJLS). 
Note that the estimators of \citet{semenova2025debiased} and \citet{jiModelAgnosticCovariateAssistedInference2026} coincide for the DTE when using a known propensity score, and that CAIDE by construction estimates bounds that are always (weakly) narrower than SJLS (Appendix \ref{appendix.comparison}).

For $n \in \{500, 2000\}$, I simulate $10,000$ times from the following data-generating process (DGP):
\begin{itemize}
    \item $d = 20$ covariates:
    \begin{itemize}[label=$\circ$]
      \item $X \sim \mathcal{N}(0, \Sigma)$, $\Sigma_{ij} = 0$ if exactly one of $i$ or $j \in \{1,2\}$, $\Sigma_{ij} = 0.5^{|i-j|}$ otherwise.
    \end{itemize}
    \item Potential outcomes, with $G_1 = \{1, 2\}$ and $G_2 = \{15, \dots, 20\}$:
    \begin{itemize}[label=$\circ$]
      \item $\beta_0 = (3^{1}, 3^{0}, 0, 0, \dots, 0, 3^{-1}, 3^{-2}, \dots, 3^{-6})$; $\Theta_{0, ij} = 0.2 (-1)^{\max(i,j)+1}$ if $i, j \in G_1$ or $i, j \in G_2$, and $\Theta_{0, ij} = 0$ otherwise;
      \item $\beta_{\tau, i} = 1$ if $i \in G_1 \cup G_2$, and $\beta_{\tau, i} = 0$ otherwise; $\Theta_{\tau, ij} = 0.2$ if $i, j \in G_1$ or $i, j \in G_2$, and $\Theta_{\tau, ij} = 0$ otherwise;
      \item $Y_i(0) = X'_i \beta_0 + X'_i \Theta_0 X_i$;
      \item $Y_i(1) = Y_i(0) -1 + X'_i \beta_{\tau} + X'_i \Theta_{\tau} X_i$;
      \item $Y_i = D_i Y_i(1) + (1 - D_i) Y_i(0)$, $D_i \sim \text{Bernoulli}(0.5)$.
    \end{itemize}
\end{itemize}

For this DGP, the true value is $\theta_0 \approx 0.43$ (computed numerically). 
I fit with $K=5$ and $\delta=0$ the models Quantile Neural Networks, Support Vector Machine (Regression), and No Covariates ($s(x)=0$). 
I also calculate an Oracle model that uses the sharp transformations defined in Theorem \ref{th.best_s}, which are generally unknown if the DGP is unknown. 
Technical details are deferred to Online Appendix \ref{appendix.ml}. 
I consider two scenarios: (i) $p=10$, where only covariates $X_1$ through $X_{10}$ are observed; (ii) $p=20$, where all $20$ covariates are observed. 
Note that when $p=20$, $\theta$ is point identified since potential outcomes are uniquely determined by all $20$ covariates. 

\begin{figure}[ht]
    \centering
    \begin{minipage}{0.95\textwidth}
        \includegraphics[width=\linewidth]{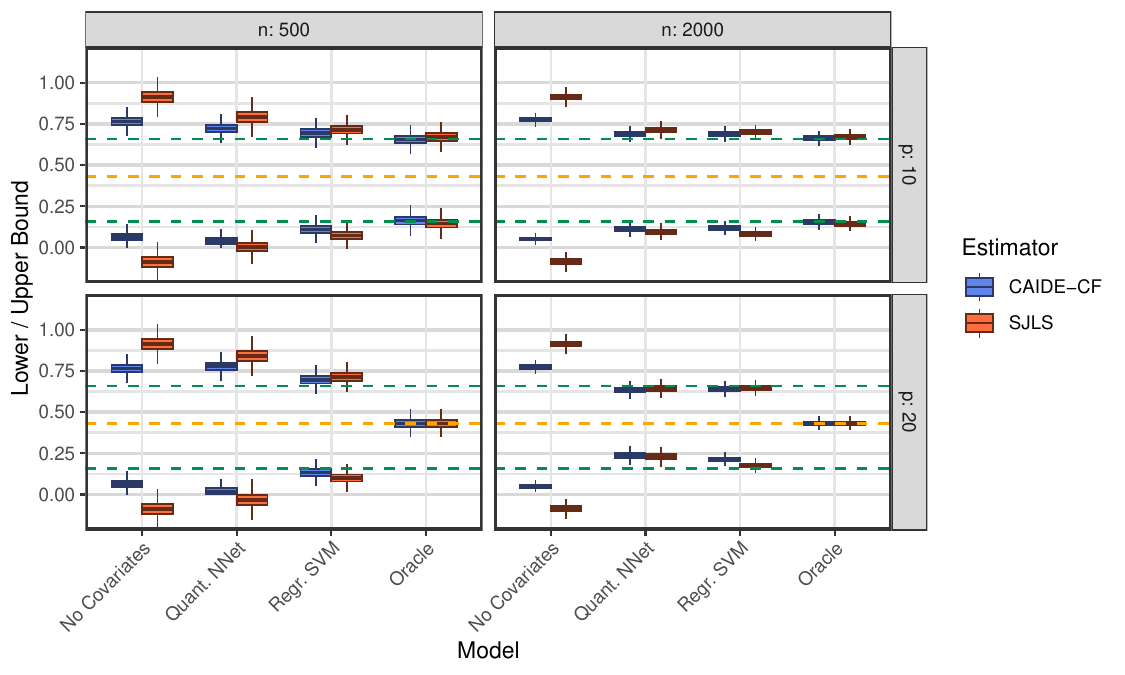}    
        {\footnotesize \textbf{Note:} Green dashed lines represent the sharp identified set of $\theta$ when $p=10$, and orange dashed lines represent the true value of $\theta$. 
        Boxplots consist of the median, two hinges at the first and third quartiles, and whiskers extending to the most extreme data points no further than 1.5 times the interquartile range (third minus first quartile) from the hinge.}
    \end{minipage}
    \caption{Distribution of Estimators for Lower and Upper Bounds on $\theta$}
    \label{fig.sim}
\end{figure}

\begin{table}[ht]
    \caption{Measures of power and size for lower bound on $\theta$}
    \label{tab.sim}
    \centering
    \begin{tabular}{cccccccc}
        \toprule
        &  & \multicolumn{3}{c}{p = 10} & \multicolumn{3}{c}{p = 20} \\ 
        \cmidrule(lr){3-5} \cmidrule(lr){6-8}
        &  & \multicolumn{2}{c}{Reject $H_0$:} &  & \multicolumn{2}{c}{Reject $H_0$:} &  \\ 
        \cmidrule(lr){3-4} \cmidrule(lr){6-7}
        n & Estimator & $\theta = 0$ & $\theta \ge \theta_0$ & Avg Length & $\theta = 0$ & $\theta \ge \theta_0$ & Avg Length \\ 
        \midrule
        \multicolumn{8}{l}{No Covariates} \\ 
        \midrule
        $500$ & CAIDE-SS & $0.000$ & $0.000$ & $0.998$ & $0.000$ & $0.000$ & $0.998$ \\ 
        $500$ & CAIDE-CF & $0.722$ & $0.000$ & $0.802$ & $0.722$ & $0.000$ & $0.802$ \\ 
        $500$ & SJLS & $0.000$ & $0.000$ & $1.000$ & $0.000$ & $0.000$ & $1.000$ \\ 
        $2,000$ & CAIDE-SS & $0.001$ & $0.000$ & $0.892$ & $0.001$ & $0.000$ & $0.892$ \\ 
        $2,000$ & CAIDE-CF & $0.983$ & $0.000$ & $0.776$ & $0.983$ & $0.000$ & $0.776$ \\ 
        $2,000$ & SJLS & $0.000$ & $0.000$ & $0.965$ & $0.000$ & $0.000$ & $0.965$ \\ 
        \midrule
        \multicolumn{8}{l}{Regr. SVM} \\ 
        \midrule
        $500$ & CAIDE-SS & $0.004$ & $0.000$ & $0.935$ & $0.007$ & $0.000$ & $0.948$ \\ 
        $500$ & CAIDE-CF & $0.970$ & $0.000$ & $0.698$ & $0.998$ & $0.000$ & $0.676$ \\ 
        $500$ & SJLS & $0.806$ & $0.000$ & $0.815$ & $0.964$ & $0.000$ & $0.791$ \\ 
        $2,000$ & CAIDE-SS & $0.414$ & $0.000$ & $0.804$ & $0.994$ & $0.000$ & $0.716$ \\ 
        $2,000$ & CAIDE-CF & $1.000$ & $0.000$ & $0.628$ & $1.000$ & $0.000$ & $0.485$ \\ 
        $2,000$ & SJLS & $1.000$ & $0.000$ & $0.707$ & $1.000$ & $0.000$ & $0.562$ \\ 
        \midrule
        \multicolumn{8}{l}{Oracle} \\ 
        \midrule
        $500$ & CAIDE-CF & $0.999$ & $0.000$ & $0.608$ & $1.000$ & $0.051$ & $0.103$ \\ 
        $2,000$ & CAIDE-CF & $1.000$ & $0.000$ & $0.568$ & $1.000$ & $0.049$ & $0.052$ \\ 
        \bottomrule
    \end{tabular}
\end{table}

Figure \ref{fig.sim} shows the distribution of the estimators CAIDE-CF (Cross Fitting) and SJLS for the lower and upper bounds on $\theta$, for $n=500,2000$ and scenarios $p=10,20$. 
It illustrates that the estimators are centered around bounds determined by the different models for estimating $\hat{s}$. 
Some models yield substantially narrower bounds than others. 
For example, in most cases, the Regr. SVM model led to the narrowest bounds, close to the sharp bounds in the $p=10$ case. 
A larger sample size leads to smaller variation in the estimators. Observing all $p=20$ covariates, instead of $p=10$, can lead to narrower bounds (e.g., Regr. SVM and $n=2000$) but also to wider bounds if the sample is smaller (Quant. NNet and $n=500$). 
Throughout, the estimates of CAIDE are narrower than those of SJLS.

In Table \ref{tab.sim}, I show the probability of rejecting the hypotheses $H_0: \theta = 0$ and $H_0: \theta \ge \theta_0$, using the one-sided confidence interval based on the lower bound in Theorems \ref{th.fs_cis} and \ref{th.cis} for the sample-splitting and cross-fitting estimators, respectively. 
Here, $\theta_0 \approx 0.43$ is the true value of $\theta$ for this DGP. 
A high rejection rate for the first hypothesis indicates power since $0$ is not in the sharp identified set of $\theta$. 
The rejection rate of the second hypothesis indicates size, and it should ideally be less than or equal to the nominal $0.05$. 
I also measure the average length of the confidence intervals, defined as the distance between the lower and upper bounds calculated from the one-sided CIs in Theorems \ref{th.fs_cis} and \ref{th.cis}. 
I compare $n=500,2000$ and $p=10,20$ for different models and estimators. 
CAIDE-SS is the sample-splitting estimator of Section \ref{section.finite_sample}, and CAIDE-CF is the cross-fitting estimator of Section \ref{section.asymptotics}. 
The table shows that including covariates, having a larger sample size, or observing all $20$ covariates increases power while preserving size. 
The sample-splitting estimator are conservative and exhibit the largest average length for Regr. SVM. 
Nonetheless, they retain meaningful power. 
For example, it excludes zero in $99.4\%$ of the cases when $p=20$, $n=2000$, and the model is Regr. SVM. 
The cross-fitting estimator performs better than SJLS in terms of smaller confidence intervals while preserving coverage. 
It achieves the nominal coverage rate for the oracle model when $p=20$. 

\section{Application to Microcredit} \label{section.application}

In 2022, microcredit reached more than 170 million borrowers worldwide, with a total loan portfolio of over \$180 billion \citep{Convergences2023ImpactFinance}. 
Theory suggests that microborrowing can be beneficial or harmful (\citealp{banerjee2013microcredit}, \citealp{garz2021consumer}), and evidence suggests that average effects are small if not negligible (\citealp{banerjee2015six}, \citealp{meager2019understanding}), and analyses of heterogeneity find null effects over most of the distribution of household profits, with large but imprecisely estimated positive effects in its upper tail (e.g., \citealp{meager2022aggregating}).

Yet, statistically significant evidence of negative effects has been limited.
\cite{vivaltHowMuchCan2020}, for example, mentions only one study with a statistically significant negative effect across all outcomes considered: one study for profits, and zero studies for each of assets, consumption, probability of owning a business, savings, and total income (Table B.1 of \cite{vivaltHowMuchCan2020}'s online appendix).
\citet{attanasioImpactsMicrofinanceEvidence2015} find a negative average treatment effect on business profits that is statistically significant at the 10\% level, but not if corrected for multiple-hypothesis testing. 
Other works document negative quantile treatment effects (QTE), for example \citet{crepon2015estimating} on business profits at the 10th percentile\footnote{\citet{crepon2015estimating} interpret the finding as a ``reduced form'' that may not capture the heterogeneity of credit access itself, since take-up is low (randomization was conducted at the village level) and ``we do not know where the compliers lie in the distribution of outcomes''. They also suggest the negative effect ``might be partially due to long-term investments misclassified as current expenses'', supported by a positive QTE on consumption at the same percentile.}  and \citet{banerjeeMiracleMicrofinanceEvidence2015} at the median at the 10\% significance level. 
However, the QTE only recovers the distribution of treatment effects under the strong assumption of rank-preservation, which rules out, for example, that two individuals with the same value of $Y(0)$ may have different values of $Y(1)$. 
Otherwise, QTEs simply recover the difference in quantiles between the distributions of $Y(1)$ and $Y(0)$ \citep{firpo2007efficient}, which may or may not have a causal interpretation. 

I revisit five studies that randomized microcredit at the individual level, and find evidence of both positive and negative treatment effects. 
I focus on individual-level randomization because it yields higher take-up rates \citep{banerjee2015six}, which facilitates prediction of treatment effects.
Four of the five studies randomized being offered a loan, while the Egypt study \citep{bryanBigLoansSmall2024} randomized the loan amount, with the treated group offered twice as much as the control group.
In both cases, the treatment represents an increase in credit access, since individuals often borrow from multiple sources simultaneously.
Table \ref{tab.datasets} presents a brief description of the datasets. 
An important note is that my goal is not to conduct a meta-analysis or draw general conclusions about microcredit's effects across contexts, but to illustrate CAIDE's potential to detect important distributional effects—including statistically significant negative impacts—that existing analyses have not previously documented. 
I estimate $\theta$ using household income as the outcome, where $\theta$ is the fraction of individuals worse off from being offered a microloan (or a larger loan in Egypt).
I focus the analysis on household income, as it aggregates earnings across all income sources and is therefore less sensitive to reallocation of labor across activities. 
I report results for business profits in Online Appendix \ref{appendix.additional_tables}.

CAIDE provides informative bounds even in this challenging context. 
First, the setting is challenging because the distributions of income/profits are similar between the treated and control groups. 
This similarity is reflected in not statistically significant ATEs in Table \ref{tab.datasets}, and estimated $[\theta_L, \theta_U]$ close to $[0,1]$ if covariates are not considered (Table \ref{tab.hhinc}). 
Second, the sample sizes are relatively small (ranging from 548 to 1964), with a large number of covariates (from 24 to 97). 
Finally, predicting business outcomes is itself a difficult task \citep{mckenzie2019predicting}.

\begin{table}[ht]
\centering
\begin{threeparttable}
\caption{Description of datasets in microcredit}
\label{tab.datasets}
\renewcommand{\arraystretch}{1.5}
\begin{tabular}{@{}p{0.32\linewidth}*{5}{c}@{}}
    \toprule
    Country & Bosnia & Egypt & Philippines (1) & Phil. (2) & S. Africa \\
    \midrule
    Year of implementation & 2009 & 2016 & 2006 & 2010 & 2004 \\
    Treated/control obs. & 551/444 & 477/478 & 891/222 & 1651/313 & 235/313 \\
    Number of covariates & 48 & 97 & 24 & 39 & 33 \\
    Avg. loan size in treatment group (USD PPP) & \$1,389 & \$9,912 & \$1,045 & \$492 & \$180 \\
    Avg household income at endline (USD PPP) & \$21,735 & \$26,844 & \$12,173 & \$10,280 & \$14,845 \\
    \begin{tabular}[c]{@{}c@{}}ATE in Household Income\\(t-statistic)\end{tabular} & \begin{tabular}[c]{@{}c@{}}\$609\\ (0.593)\end{tabular} & \begin{tabular}[c]{@{}c@{}}\$1,224\\ (0.573)\end{tabular} & \begin{tabular}[c]{@{}c@{}}\$1,374\\ (0.819)\end{tabular} & \begin{tabular}[c]{@{}c@{}}\$110\\ (0.103)\end{tabular} & \begin{tabular}[c]{@{}c@{}}-\$520\\ (-0.159)\end{tabular} \\
    \bottomrule
\end{tabular}
\begin{tablenotes}
\small
\item Study citations respectively: \citet{augsburg2015impacts}, \citet{bryanBigLoansSmall2024}, \citet{karlan2011microcredit}, \citet{karlan2016follow}, \citet{karlan2010expanding}.
\end{tablenotes}
\end{threeparttable}
\end{table}

Using $K = 10$, I fit the estimators of Section \ref{section.asymptotics}\footnote{The sample-splitting estimators of Section \ref{section.finite_sample} with a 50-50 split yield confidence intervals equal to $[0,1]$ at the $10\%$ significance level. In these datasets, the critical values in Theorem \ref{th.fs_cis} are $c_{0.10} \approx 0.156, 0.158, 0.174, 0.140, 0.211$ for the one-sided CIs and $c_{0.05} \approx 0.173, 0.176, 0.193, 0.156, 0.234$ for the two-sided CIs, respectively, in the order of Table \ref{tab.datasets}. 
In this setting, one may need larger samples to obtain informative finite-sample CIs: with a fixed treatment proportion, quadrupling the sample size halves these critical values.}. For each fold, I fit another layer of sample splitting to pick via 10-fold cross-validation one of the models: Quantile Forests, Quantile Neural Networks, Random Forest (Regression), Neural Networks (Reg.), Elastic Net (Reg.), Extreme Gradient Boosting (Reg.), and Constant ($s(X) = 0$). Technical details are provided in Online Appendix \ref{appendix.ml}.

\begin{table}[ht]
    \centering
    \begin{threeparttable}
    \caption{Bounds on fraction harmed by microcredit - household income}
    \label{tab.hhinc}
    \begin{tabular}{@{}l|ccc|ccc@{}}
    \toprule
    & \multicolumn{3}{c}{Lower Bound} & \multicolumn{3}{c}{Upper Bound} \\
    \cline{2-4} \cline{5-7}
    & CAIDE & No Covariates & SJLS & CAIDE & No Covariates & SJLS \\ \midrule
    Bosnia & 
    \begin{tabular}[c]{@{}c@{}}0.108\\(0.030)\\{[}0.000{]}\end{tabular} & 
    \begin{tabular}[c]{@{}c@{}}0.040\\(0.027)\\{[}0.073{]}\end{tabular} & 
    \begin{tabular}[c]{@{}c@{}}0.024\\(0.036)\\{[}0.254{]}\end{tabular} & 
    \begin{tabular}[c]{@{}c@{}}0.946\\(0.028)\\{[}0.028{]}\end{tabular} & 
    \begin{tabular}[c]{@{}c@{}}0.973\\(0.031)\\{[}0.195{]}\end{tabular} & 
    \begin{tabular}[c]{@{}c@{}}0.977\\(0.047)\\{[}0.310{]}\end{tabular} \\ \midrule
    Egypt & 
    \begin{tabular}[c]{@{}c@{}}0.260\\(0.022)\\{[}0.000{]}\end{tabular} & 
    \begin{tabular}[c]{@{}c@{}}0.009\\(0.014)\\{[}0.263{]}\end{tabular} & 
    \begin{tabular}[c]{@{}c@{}}0.248\\(0.033)\\{[}0.000{]}\end{tabular} & 
    \begin{tabular}[c]{@{}c@{}}0.787\\(0.022)\\{[}0.000{]}\end{tabular} & 
    \begin{tabular}[c]{@{}c@{}}0.966\\(0.021)\\{[}0.051{]}\end{tabular} & 
    \begin{tabular}[c]{@{}c@{}}0.792\\(0.033)\\{[}0.000{]}\end{tabular} \\ \midrule
    Philippines (1) & 
    \begin{tabular}[c]{@{}c@{}}0.083\\(0.039)\\{[}0.018{]}\end{tabular} & 
    \begin{tabular}[c]{@{}c@{}}0.071\\(0.039)\\{[}0.036{]}\end{tabular} & 
    \begin{tabular}[c]{@{}c@{}}0.077\\(0.053)\\{[}0.075{]}\end{tabular} & 
    \begin{tabular}[c]{@{}c@{}}0.971\\(0.019)\\{[}0.057{]}\end{tabular} & 
    \begin{tabular}[c]{@{}c@{}}0.931\\(0.027)\\{[}0.005{]}\end{tabular} & 
    \begin{tabular}[c]{@{}c@{}}1.048\\(0.049)\\{[}1.000{]}\end{tabular} \\ \midrule
    Philippines (2) & 
    \begin{tabular}[c]{@{}c@{}}0.145\\(0.033)\\{[}0.000{]}\end{tabular} & 
    \begin{tabular}[c]{@{}c@{}}0.022\\(0.018)\\{[}0.103{]}\end{tabular} & 
    \begin{tabular}[c]{@{}c@{}}0.075\\(0.034)\\{[}0.013{]}\end{tabular} & 
    \begin{tabular}[c]{@{}c@{}}0.739\\(0.032)\\{[}0.000{]}\end{tabular} & 
    \begin{tabular}[c]{@{}c@{}}0.951\\(0.032)\\{[}0.062{]}\end{tabular} & 
    \begin{tabular}[c]{@{}c@{}}0.751\\(0.047)\\{[}0.000{]}\end{tabular} \\ \midrule
    South Africa & 
    \begin{tabular}[c]{@{}c@{}}0.028\\(0.017)\\{[}0.053{]}\end{tabular} & 
    \begin{tabular}[c]{@{}c@{}}0.033\\(0.018)\\{[}0.036{]}\end{tabular} & 
    \begin{tabular}[c]{@{}c@{}}-0.017\\(0.054)\\{[}1.000{]}\end{tabular} & 
    \begin{tabular}[c]{@{}c@{}}0.876\\(0.043)\\{[}0.002{]}\end{tabular} & 
    \begin{tabular}[c]{@{}c@{}}0.913\\(0.043)\\{[}0.021{]}\end{tabular} & 
    \begin{tabular}[c]{@{}c@{}}0.938\\(0.049)\\{[}0.105{]}\end{tabular} \\
    \bottomrule
    \end{tabular}
    \begin{tablenotes}
    \small
    \item Standard errors in parentheses, and p-values in brackets. The null hypothesis for the lower bound is $\theta^*_L = 0$, and for the upper bound $\theta^*_U = 1$. Response variable: $\log(1 + household\ income )$.
    \end{tablenotes}
    \end{threeparttable}
\end{table}

Table \ref{tab.hhinc} shows the estimates for the lower and upper bounds on $\theta$ for $\delta = 0$ using CAIDE with and without covariates\footnote{Note that the estimators and one-sided CIs without covariates ($s(x)=0$) are equivalent to the ones in \citet{fan2010sharp}.}. 
The outcome is the logarithm of annual household income.\footnote{Since $y \mapsto \log(y+1)$ is strictly increasing for $y \ge 0$, the event $\{Y(1) - Y(0) \le 0\}$ is equivalent to $\{\log(Y(1)+1) - \log(Y(0)+1) \le 0\}$, so using the log-transformed outcome leaves $\theta$ unchanged for $\delta = 0$.} 
Assuming the distribution of $Y(1) - Y(0)$ has no mass on zero, $\theta = P(Y(1) - Y(0) \le 0) = P(Y(1) - Y(0) < 0)$, the lower bound on $\theta$ is the minimum proportion of households that have a negative effect from the microloan, and the upper bound is the maximum proportion. 
Note that one minus the upper bound denotes the minimum proportion of households with a positive effect. 
Results for $\delta = -0.05$ and $\delta = 0.05$ are similar and are presented in Online Appendix Table \ref{tab.deltas}. 
In Bosnia, for example, I estimate that at least $10.8\%$ of the households would be better off without the microloan (p-value $ < 10^{-3}$), and at most $94.6\%$ have negative effects (p-value $0.028$). 
That is, at least $5.4\%$ are better off ($1 - 0.946$). 
Overall, I find significant lower bounds at the $10\%$ level for all five datasets and at the $5\%$ level for four of them. 
As for the upper bound, I find statistical significance at $5\%$ in four datasets. These results reveal an important presence of heterogeneity, even when the ATEs are not statistically significant (Table \ref{tab.datasets}).

Incorporating covariates makes a substantial difference in many of the point estimates. 
In Egypt, for example, the lower bound without covariates is $0.009$ (p-value $0.263$), whereas incorporating covariates improves the estimate to $0.260$ (p-value $<10^{-3}$). 
The differences between the point estimates with and without covariates are significant at the $5\%$ level for both the lower and upper bounds for Egypt and the Philippines (2). 

Table \ref{tab.hhinc} also illustrates how CAIDE can outperform SJLS (\citealp{semenova2025debiased}, \citealp{jiModelAgnosticCovariateAssistedInference2026}) in terms of narrower confidence intervals: of the ten bounds in the table, CAIDE rejects the null at the 10\% level for all ten, and SJLS for five.
In the case of the lower bound for Bosnia, for example, it is estimated at $10.8\%$ with CAIDE vs. $2.4\%$ with SJLS, leading to a substantial difference in p-values ($<10^{-3}$ vs. $0.254$). 
It also illustrates how SJLS can estimate negative lower bounds (South Africa), or upper bounds greater than one (Philippines (1)). 
In both cases, the estimates with CAIDE are statistically significant at the $10\%$ level.

Table \ref{tab.bprofits} in Online Appendix \ref{appendix.additional_tables} shows results for business profits. Again, including covariates improves the bounds. Averaging estimates across datasets, I find that at least $10.8\%$ (p-value $< 10^{-3}$) of the population was harmed, and at least $9.8\%$ (p-value $< 10^{-3}$) benefited from microcredit.

This evidence does not, however, speak to the nature or magnitude of the negative effects. 
For example, the negative outcomes could be relatively small and a natural consequence of investments being risky. 
In that case, the decision to take a microloan could still be rational if the prospect of gains is better than the risk of loss.

\section{Conclusion} \label{section.conclusion}

I propose two novel inference approaches to the distributional treatment effect $\theta = P ( Y(1) - Y(0) \le \delta )$ for any fixed $\delta$. 
A new characterization of the sharp identified set of $\theta$ is provided, based on covariate adjustment. 
It has the desirable property of being robust to the choice of covariate-adjustment term $s(x)$. 
If this term is misspecified, the identified set is wider but still valid. 
As a consequence, any method can be used to estimate $s$, including machine learning algorithms. 
I provide finite-sample valid inference using sample-splitting and asymptotic inference using cross-fitting, which is also asymptotically exact under additional conditions. 

I illustrate the practical relevance of the method in a simulation study and an application to microcredit. 
The application reveals an important presence of heterogeneity, with statistically significant proportions of individuals with positive and negative treatment effects. 
The results show that the method is able to provide informative bounds even in challenging contexts, such as when the ATE is not statistically significant, the sample size is small, and the number of covariates is large. 

By using CAIDE, an applied researcher can learn more about heterogeneous effects not only by gathering larger samples but also by collecting information on variables that help predict potential outcomes. %

\clearpage

\begin{appendices}

\counterwithin{lemma}{section}
\counterwithin{proposition}{section}
\counterwithin{theorem}{section}

\section{}

\subsection{Comparison with \texorpdfstring{\citet{semenova2025debiased}}{Semenova} and \texorpdfstring{\citet{jiModelAgnosticCovariateAssistedInference2026}}{Ji et al.}} \label{appendix.comparison}

I compare the asymptotic properties of CAIDE with the estimators proposed by \citet{semenova2025debiased} and \citet{jiModelAgnosticCovariateAssistedInference2026} (SJLS). 
I focus on the cross-fitting version of my estimators (Definition \ref{def.theta_hat_cf}) and compare them to the cross-fitting estimators in SJLS. %
Note that in the case of the DTE, $\theta$, the estimators of \citet{semenova2023adaptive} and \citet[Example 4]{jiModelAgnosticCovariateAssistedInference2026} are equivalent.\footnote{\citet{semenova2025debiased}, the current version of \citet{semenova2023adaptive}, generalizes this estimator to doubly robust scores but covers the DTE only for outcomes with finite support.}

Since SJLS considered a known propensity score, I compare SJLS to a version of CAIDE that uses the true propensity score (Online Appendix \ref{appendix.relaxing_prop_score}) instead of using the fraction of treated units in the sample. 
In Online Appendix \ref{appendix.prop_score}, I show that using the actual number of treated units (as in Definition \ref{def.theta_hat_cf}) instead of the known propensity score leads to a smaller asymptotic variance. 
Hence, using the true propensity score in CAIDE is the natural comparison. 
I focus on the lower bound for simplicity, and analogous results hold for the upper bound. 
Define 
$$\hat{\Delta}(t, \hat{s}_L) = \frac{1}{n} \sum_{i = 1}^n \biggl( \frac{D_i}{p(X_i)} \mathbb{I}(Y_i \le \hat{s}_{L, k(i)}(X_i) + t) - \frac{1 - D_i}{1 - p(X_i)} \mathbb{I}(Y_i \le \hat{s}_{L, k(i)}(X_i) + t) \biggr),$$
where $\hat{s}_L = (\hat{s}_{L, 1}, \dots, \hat{s}_{L, K})$. 

Then, CAIDE is equivalent to $\hat{\theta}_C := \max_t \hat{\Delta}(t, \hat{s}_L)$ and SJLS is given by $\hat{\theta}_S := \hat{\Delta}(0, \hat{s}_L)$. 
Theorem \ref{th.asymptotic_dist_simplified} of this paper is comparable to Theorem 1 in \citet{semenova2025debiased}, and Theorems 3.1 and 3.4 in \citet{jiModelAgnosticCovariateAssistedInference2026}. 
They all show asymptotic validity of confidence intervals. 
Define confidence intervals as $CI_C$ for CAIDE and $CI_S$ for SJLS. 
Since by construction $\hat{\theta}_C \ge \hat{\theta}_S$, it follows that $\hat{\theta}_C$ first-order stochastically dominates $\hat{\theta}_S$. 
Let $\theta_{n}^* = \Delta_{P_n}(0, s^*_L)$ be the sharp lower bound as in (\ref{def.theta_sharp}) and Theorem \ref{th.best_s}. 
Under the conditions of these theorems, it follows that, for any sequence $\theta_{n} < \theta_{n}^*$, 
\begin{equation} \label{eq.power_dominance}
    \liminf_{n \to \infty}  P_n(\theta_{n} \in CI_S) - P_n(\theta_{n} \in CI_C) \ge 0.
\end{equation}

Furthermore, consider a scenario where $\max_t \Delta_{P_n}(t, s_{L,P}) > \Delta_{P_n}(0, s_{L,P}) + M$ for some $M > 0$, where $s_{L,P}$ is the limit of $\hat{s}_{L,k}$ (as in A\ref{as.limit_of_s}). 
This is usually the case when $s_{L,P}$ is not a sharp transformation $s^*_{L,P}$. 
Then, for any sequence $\theta_n$ satisfying
\[
\Delta_{P_n}(0,s_{L,P_n})+M
\leq\theta_n
\leq
\max_t\Delta_{P_n}(t,s_{L,P_n})-\varepsilon
\]
for some fixed $\varepsilon>0$, (\ref{eq.power_dominance}) holds with strict inequality.
Indeed, 
$P_n(\theta_n\in CI_S)\to1$ and $P_n(\theta_n\in CI_C)\to0$.

Since $\theta_{n} < \theta_{n}^*$, a lower coverage probability translates into higher power. 
Hence, $CI_C$ dominates $CI_S$ in terms of power.
Note that validity of $CI_C$ relies on A\ref{as.uniqueness_t_infty}, on the uniqueness of the maximizer of $\Delta_{P_n}(t, s_{L,P})$, which is not assumed in SJLS.\footnote{For their exactness result for the DTE, \citet{jiModelAgnosticCovariateAssistedInference2026} assume a stronger condition than A\ref{as.uniqueness_t_infty}: the conditional CDF differences are single-peaked at each covariate value, both for the true and the estimated distributions.} 
Online Appendix \ref{appendix.relaxing_maximizer} discusses inference without it. 

Note the two approaches start from different principles. 
SJLS begin from the dual of an optimization problem, while CAIDE begins from the invariance of $\theta$ to covariate adjustment and reduces the problem to one involving only two marginal distributions. 
For the DTE, the two routes lead to cross-fitted estimators of a common form, and the comparison above shows that SJLS coincides with CAIDE evaluated at $t = 0$, the optimizer when the first stage is already a sharp transformation. 
The maximization over $t$ is the sample analogue of the covariate-adjustment construction, and it is what delivers the finite-sample interval of Section~\ref{section.finite_sample}, the power ranking shown here, the power improvement in the application of Section~\ref{section.application}, and the exactness result under continuous outcomes without a margin condition.

\subsection{Definitions Omitted from the Main Text} \label{appendix.definitions}

Let $t_{L,P} = t_{max,P}$; 
$t_{U,P} = t_{min,P}$; 
$\pi = P(D = 1)$ (as in A\ref{as.rct}); 
for $A \in \{L, U\}$ and $j\in \{0, 1 \}$, define $Z_{A, P}(j) = \mathbb{I}(Y(j) \le s_{A,P}(X) + t_{A,P})$;
$\sigma^2_{A,P}(j) = \Var[Z_{A,P}(j)]$. 
Then, for $A \in \{L, U\}$:
$$\sigma^2_{A,P} = \frac{\Var[Z_{A,P}(1)]}{\pi} + \frac{\Var[Z_{A,P}(0)]}{1 - \pi}$$
and
$$\sigma_{L,U,P} = \frac{\Cov[Z_{L,P}(1), Z_{U,P}(1)]}{\pi} + \frac{\Cov[Z_{L,P}(0), Z_{U,P}(0)]}{1 - \pi}$$

For the sample analogues, for $A \in \{L, U\}$ define 
$$\hat{F}_{A,j}(t) = \frac{1}{n_j} \sum_{i : D_i = j} \mathbb{I}(Y_i \le \hat{s}_{A, k(i)}(X_i) + t);$$ 
$$\hat{t}_{L} = \arg \max_t \hat{F}_{L,1}(t) - \hat{F}_{L,0}(t); \ \hat{t}_{U} = \arg \min_t \hat{F}_{U,1}(t) - \hat{F}_{U,0}(t);$$ 
$$\hat{\sigma}^2_{A,j} = \frac{1}{n_j} \sum_{i : D_i = j} \left[ \mathbb{I}(Y_i \le \hat{s}_{A, k(i)}(X_i) + \hat{t}_{A}) - \hat{F}_{A,j}(\hat{t}_{A}) \right]^2;$$ 
\begin{align*}
    & \hat{\sigma}_{L,U,j} = \\
    & \frac{1}{n_j} \sum_{i : D_i = j} \left[ \mathbb{I}(Y_i \le \hat{s}_{L, k(i)}(X_i) + \hat{t}_{L}) - \hat{F}_{L,j}(\hat{t}_{L}) \right] \left[ \mathbb{I}(Y_i \le \hat{s}_{U, k(i)}(X_i) + \hat{t}_{U}) - \hat{F}_{U,j}(\hat{t}_{U}) \right]; 
\end{align*}
and $\hat{\pi} = n_1 / n$. Then, for $A \in \{L, U\}$,
$$\hat{\sigma}^2_{A} = \frac{\hat{\sigma}^2_{A,1}}{\hat{\pi}} + \frac{\hat{\sigma}^2_{A,0}}{1 - \hat{\pi}}; \quad \hat{\sigma}_{L,U} = \frac{\hat{\sigma}_{L,U,1}}{\hat{\pi}} + \frac{\hat{\sigma}_{L,U,0}}{1 - \hat{\pi}}.$$

\subsection{Proofs of the Main Results} \label{appendix.proofs}
\begin{proof}[\hypertarget{proof.th.fs_cis}{Proof of Theorem} \ref{th.fs_cis}] \, \\
    For the case of the lower bound, denote for $j=0,1$
        $$\hat{F}_{j, L}(t) = \frac{1}{|\mathcal{I}^{j}_M|} \sum_{i \in \mathcal{I}^{j}_M} \mathbb{I} \left(Y_i - \hat{s}_{L}(X_i) \le t \right)$$
        Let $\mathcal{I}_{A}$ denote the auxiliary sample, and let $F_{j, L}(t) = \mathbb{E}[\hat{F}_{j, L}(t) | \mathcal{I}_{A}]$ for $j=0,1$ denote the true CDF of $Y(j) - \hat{s}_L(X)$ taking $\hat{s}_L$ as fixed. 
        Let $\hat{t}_{max} \in \arg\max_t [\hat{F}_{1, L}(t) - \hat{F}_{0, L}(t)]$.
    
        Then, we have that $\widetilde{\theta}_L = \hat{F}_{1, L}(\hat{t}_{max}) - \hat{F}_{0, L}(\hat{t}_{max})$. 
        Finally, denote $c_{\alpha,1} = \left( \frac{\log(2 / \alpha)}{2 |\mathcal{I}^{1}_M|} \right)^{1/2}$, $c_{\alpha,0} = \left( \frac{\log(2 / \alpha)}{2 |\mathcal{I}^{0}_M|} \right)^{1/2}$. Note that $c_{\alpha,1} + c_{\alpha,0} = c_\alpha$ as defined in Theorem \ref{th.fs_cis}.
        
        I first focus on the case of the one-sided interval $P \left( \theta \in \left[\widetilde{\theta}_L - c_\alpha, 1 \right] \right)$. 
        Note that $F_{1, L}(t) - F_{0, L}(t)$ is a valid lower bound by (\ref{def.induced_bounds}) for any $t$ and $\hat{s}_{L}$. 
        Thus, $\theta \in [F_{1, L}(\hat{t}_{max}) - F_{0, L}(\hat{t}_{max}), 1]$ and $F_{1, L}(\hat{t}_{max}) - F_{0, L}(\hat{t}_{max}) \le \theta$. 
        Then,
        \begin{align}
            & P \left( \theta \in \left[\widetilde{\theta}_L - c_\alpha, 1 \right] \right) \nonumber \\
            & \ge P \left( F_{1, L}(\hat{t}_{max}) - F_{0, L}(\hat{t}_{max}) \ge \widetilde{\theta}_L - c_\alpha \right) \nonumber \\
            & = 1 - P \left( \left[ \hat{F}_{1,L}(\hat{t}_{max}) - F_{1,L}(\hat{t}_{max}) \right] + \left[ F_{0,L}(\hat{t}_{max}) - \hat{F}_{0,L}(\hat{t}_{max}) \right] > c_\alpha \right) \nonumber \\
            & \ge 1 - P \left( \left[ \hat{F}_{1,L}(\hat{t}_{max}) - F_{1,L}(\hat{t}_{max}) \right] > c_{\alpha,1} \lor \left[ F_{0,L}(\hat{t}_{max}) - \hat{F}_{0,L}(\hat{t}_{max}) \right] > c_{\alpha,0} \right) \nonumber \\
            & \ge 1 - P \left( \hat{F}_{1,L}(\hat{t}_{max}) - F_{1,L}(\hat{t}_{max}) > c_{\alpha, 1} \right) - P \left( F_{0,L}(\hat{t}_{max}) - \hat{F}_{0,L}(\hat{t}_{max}) > c_{\alpha, 0} \right) \nonumber \\
            & \ge 1 - P \left( \sup_t \left( \hat{F}_{1,L}(t) - F_{1,L}(t) \right) > c_{\alpha, 1} \right) - P \left( \sup_t \left( F_{0,L}(t) - \hat{F}_{0,L}(t) \right) > c_{\alpha, 0} \right) \nonumber \\
            & = 1 - \mathbb{E} \left[ P \left( \sup_t \left( \hat{F}_{1,L}(t) - F_{1,L}(t) \right) > c_{\alpha, 1} \Bigm| \mathcal{I}_{A} \right) \right] \nonumber \\
            & \hspace{15pt} - \mathbb{E} \left[ P \left( \sup_t \left( F_{0,L}(t) - \hat{F}_{0,L}(t) \right) > c_{\alpha, 0} \Bigm| \mathcal{I}_{A} \right) \right] \label{eq.th.fs_cis.lie} \\
            & \ge 1 - \mathbb{E} \left[ \exp{(-2 |\mathcal{I}^{1}_M| c_{\alpha, 1}^2)} \right] - \mathbb{E} \left[ \exp{(-2 |\mathcal{I}^{0}_M| c_{\alpha, 0}^2)} \right] \label{eq.th.fs_cis.dkw} \\
            & = 1 - \exp{(-2 |\mathcal{I}^{1}_M| c_{\alpha, 1}^2)} - \exp{(-2 |\mathcal{I}^{0}_M| c_{\alpha, 0}^2)} \nonumber \\
            & = 1 - \alpha \nonumber
        \end{align}
        
        (\ref{eq.th.fs_cis.lie}) holds by the law of iterated expectations, and note the expected values are with respect to $\mathcal{I}_A$. 
        (\ref{eq.th.fs_cis.dkw}) holds by the one-sided DKW inequality as in \citet{massart1990tight}. 
        Note the use of sample splitting and conditioning on the auxiliary set $\mathcal{I}_{A}$ is important since it ensures that $(Y_i, \hat{s}_{L}(X_i), D_i)_{i \in \mathcal{I}^1_{M} \cup \mathcal{I}^0_{M}}$ are i.i.d., as required by the DKW inequality. 
    
        The proof for $P \left( \theta \in \left[ 0, \widetilde{\theta}_U + c_\alpha \right] \right) \ge 1 - \alpha$ is analogous. 
        Finally, the validity of the two-sided confidence interval follows immediately:
        \begin{align*}
            P \Biggl( \theta \in \left[\widetilde{\theta}_L - c_{\alpha / 2}, \widetilde{\theta}_U + c_{\alpha / 2} \right] \Biggr) & = P \Biggl( \theta \in \left[\widetilde{\theta}_L - c_{\alpha / 2}, 1 \right] \land \theta \in \left[0, \widetilde{\theta}_U + c_{\alpha / 2} \right] \Biggr) \\
            & = 1 - P \Biggl( \theta \not\in \left[\widetilde{\theta}_L - c_{\alpha / 2}, 1 \right] \lor \theta \not\in \left[0, \widetilde{\theta}_U + c_{\alpha / 2} \right] \Biggr) \\
            & \ge 1 - P \Biggl( \theta \not\in \left[\widetilde{\theta}_L - c_{\alpha / 2}, 1 \right] \Biggr) - P \Biggl( \theta \not\in \left[0, \widetilde{\theta}_U + c_{\alpha / 2} \right] \Biggr) \\
            & \ge 1 - \alpha/2 - \alpha/2 \\
            & = 1 - \alpha
        \end{align*}
    \end{proof}

\paragraph{Definitions used in the proofs of the results of Section \ref{section.asymptotics}:}

\begin{itemize}
    \item $t_{L,P} = t_{max,P}, \ t_{U,P} = t_{min,P}$;
    \item $F_{P, j}(t,s) = P ( Y(j) \le s(X) + t )$ (note $F_{P, j}(t, s)$ is random if $t$ and/or $s$ are random);
    \item $\mathcal{I}_{j,k} = \{ i \in \{1, \dots, n \} : D_i = j \land i \in \mathcal{I}_k \}$;
    \item $\hat{F}_{j, k}(t, s) = \frac{1}{n_j / K} \sum_{i \in \mathcal{I}_{j,k}} \mathbb{I}(Y_i \le s(X_i) + t)$;
    \item $\hat{t}_{L} = \arg \max_t \frac{1}{K} \sum_{k=1}^K \left[ \hat{F}_{1, k}(t,\hat{s}_{L,k}) - \hat{F}_{0, k}(t,\hat{s}_{L,k}) \right]$; 
    \item $\hat{t}_{U} = \arg \min_t \frac{1}{K} \sum_{k=1}^K \left[ \hat{F}_{1, k}(t,\hat{s}_{U,k}) - \hat{F}_{0, k}(t,\hat{s}_{U,k}) \right]$; 
    \item $\hat{\pi} = n_1 / n$. 
\end{itemize}

Theorem \ref{th.asymptotic_dist_simplified} is restated as Theorem \ref{th.asymptotic_dist}:

\begin{theorem}[\hyperlink{proof.th.asymptotic_dist}{Asymptotic distribution of cross-fitting estimators}] \label{th.asymptotic_dist}
    Let $\mathcal{P}$ denote a set of probability measures satisfying Assumptions \ref{as.rct} and \ref{as.asymptotic}. 
    Let $\{P_n\}_{n \ge 1} \subset \mathcal{P}$ be a sequence of probability measures such that
    $$\begin{pmatrix} \sigma^2_{L,P_n} & \sigma_{L,U,P_n} \\ \sigma_{L,U,P_n} & \sigma^2_{U,P_n} \end{pmatrix} \to \begin{pmatrix} \sigma^2_L & \sigma_{L,U} \\ \sigma_{L,U} & \sigma^2_{U} \end{pmatrix}$$
    for some $\sigma^2_L, \sigma^2_U \ge 0$ and $\sigma_{L,U} \in \mathbb{R}$.
    
    Then, there exist $\bar{\theta}_{L, P_n}, \bar{\theta}_{U, P_n}$ with $\bar{\theta}_{L, P_n} \le \theta_{L,P_n}^* \le \theta_{P_n} \le \theta_{U,P_n}^* \le \bar{\theta}_{U, P_n}$ such that
    $$\sqrt{n} \begin{bmatrix} \hat{\theta}_{L} - \bar{\theta}_{L, P_n} \\ \hat{\theta}_{U} - \bar{\theta}_{U, P_n} \end{bmatrix} \overset{d}{\to} \mathcal{N} \left( \begin{bmatrix} 0 \\ 0 \end{bmatrix}, \begin{bmatrix} \sigma^2_{L} & \sigma_{L,U} \\ \sigma_{L,U} & \sigma^2_{U} \end{bmatrix} \right)$$
\end{theorem}

\begin{proof}[\hypertarget{proof.th.asymptotic_dist}{Proof of Theorem} \ref{th.asymptotic_dist}] \, \\

For $A \in \{L,U\}$, define 
\begin{itemize}
    \item $\bar{\theta}_{A, P, k} = F_{P,1}(\hat{t}_{A},\hat{s}_{A,k}) - F_{P,0}(\hat{t}_{A}, \hat{s}_{A,k})$,
    \item $\bar{\theta}_{A,P} = \frac{1}{K} \sum_{k=1}^K \bar{\theta}_{A, P, k}$, 
    \item $\hat{\theta}_{A,k} = \hat{F}_{1, k}(\hat{t}_{A},\hat{s}_{A,k}) - \hat{F}_{0, k}(\hat{t}_{A},\hat{s}_{A,k})$. 
\end{itemize}
Note that $\hat{\theta}_A = \frac{1}{K} \sum_{k=1}^K \hat{\theta}_{A,k}$, and that
$$\hat{\theta}_{A,k} - \bar{\theta}_{A, P_n, k} = \left[ \hat{F}_{1,k}(\hat{t}_{A},\hat{s}_{A,k}) - F_{P_n,1}(\hat{t}_{A},\hat{s}_{A,k}) \right] - \left[ \hat{F}_{0,k}(\hat{t}_{A},\hat{s}_{A,k}) - F_{P_n,0}(\hat{t}_{A},\hat{s}_{A,k}) \right]$$

The proof consists of showing, for $A \in \{L,U\}$ and $j=0,1$, that
\begin{align}
    & \sqrt{\frac{n_j}{K}} \left( \hat{F}_{j,k}(\hat{t}_{A}, \hat{s}_{A,k}) - F_{P_n,j}(\hat{t}_{A}, \hat{s}_{A,k}) \right) \nonumber \\
    & = \sqrt{\frac{n_j}{K}} \left( \hat{F}_{j,k}(\hat{t}_{A}, s_{A,P_n}) - F_{P_n,j}(\hat{t}_{A}, s_{A,P_n}) \right) + o_{P_n}(1) \label{eq.th.asymptotic_dist.1} \\
    & = \sqrt{\frac{n_j}{K}} \left( \hat{F}_{j,k}(t_{A,P_n}, s_{A,P_n}) - F_{P_n,j}(t_{A,P_n}, s_{A,P_n}) \right) + o_{P_n}(1) \label{eq.th.asymptotic_dist.2}
\end{align}

Finally, I show
\begin{equation} \label{eq.th.asymptotic_dist.3}
    \sqrt{n} \begin{bmatrix} \hat{\theta}_L - \bar{\theta}_{L,P_n} \\ \hat{\theta}_U - \bar{\theta}_{U,P_n} \end{bmatrix} \overset{d}{\to} \mathcal{N} \left( \begin{bmatrix} 0 \\ 0 \end{bmatrix}, \begin{bmatrix} \sigma^2_L & \sigma_{L,U} \\ \sigma_{L,U} & \sigma^2_{U} \end{bmatrix} \right)
\end{equation}

Note $\bar{\theta}_{L, P_n} \le \theta_{L,P_n}^* \le \theta_{P_n} \le \theta_{U,P_n}^* \le \bar{\theta}_{U, P_n}$ holds since, for each $k$, $\theta_{P_n} \in [\bar{\theta}_{L, P_n, k}, \bar{\theta}_{U, P_n, k}]$ holds from (\ref{def.induced_bounds}).

\begin{itemize}
    \item \textbf{Step 1: Equation (\ref{eq.th.asymptotic_dist.1})}
    
    I show that 
    \begin{align*}
        & \left( \sqrt{\frac{n_j}{K}} \right)^{-1} \sum_{i \in \mathcal{I}_{j,k}} \Bigl[ \left( \mathbb{I}(Y_i \le \hat{s}_{A,k}(X_i) + \hat{t}_{A}) - \mathbb{I}(Y_i \le s_{A,P_n}(X_i) + \hat{t}_{A}) \right) \\
        & \hspace*{90pt} - \left( F_{P_n,j}(\hat{t}_{A}, \hat{s}_{A,k}) - F_{P_n,j}(\hat{t}_{A}, s_{A,P_n}) \right) \Bigr] \\
        & =: \left( \sqrt{\frac{n_j}{K}} \right)^{-1} \sum_{i \in \mathcal{I}_{j,k}} W_i(\hat{t}_{A}) \\
        & \le \sup_t \left( \sqrt{\frac{n_j}{K}} \right)^{-1} \sum_{i \in \mathcal{I}_{j,k}} W_i(t) = o_{P_n}(1)
    \end{align*}

    The result follows from stochastic equicontinuity of $\left( \sqrt{\frac{n_j}{K}} \right)^{-1} \sum_{i \in \mathcal{I}_{j,k}} W_i(t)$ plus pointwise convergence in $t$ to zero (see, e.g., Theorem 22.9 in \citet{davidson2021stochastic}). 

    Note that
    \begin{align*}
        & \left( \sqrt{\frac{n_j}{K}} \right)^{-1} \sum_{i \in \mathcal{I}_{j,k}} \Bigl[ \left( \mathbb{I}(Y_i \le \hat{s}_{A,k}(X_i) + t) - \mathbb{I}(Y_i \le s_{A,P_n}(X_i) + t) \right) \\
        & \hspace*{90pt} - \left( F_{P_n,j}(t, \hat{s}_{A,k}) - F_{P_n,j}(t, s_{A,P_n}) \right) \Bigr] \\
        & = \left( \sqrt{\frac{n_j}{K}} \right)^{-1} \sum_{i \in \mathcal{I}_{j,k}}  \left( \mathbb{I}(Y_i \le \hat{s}_{A,k}(X_i) + t) - F_{P_n,j}(t, \hat{s}_{A,k}) \right) \\
        & \hspace*{20pt} - \left( \sqrt{\frac{n_j}{K}} \right)^{-1} \sum_{i \in \mathcal{I}_{j,k}} \left( \mathbb{I}(Y_i \le s_{A,P_n}(X_i) + t) - F_{P_n,j}(t, s_{A,P_n}) \right), \\
    \end{align*}
    and so it is enough to show stochastic equicontinuity of 
    $$\left( \sqrt{\frac{n_j}{K}} \right)^{-1} \sum_{i \in \mathcal{I}_{j,k}} \left( \mathbb{I}(Y_i \le \tilde{s}(X_i) + t) - F_{P_n,j}(t, \tilde{s}) \right)$$
    for an arbitrary random variable $\tilde{s}$ independent of $\mathcal{I}_{j,k}$. 
    This follows since $\mathcal{F} = \{ f_{t} : \mathbb{R} \rightarrow [0,1], w \mapsto f_{t}(w) = \mathbb{I} (w \le t) : t \in \mathbb{R} \}$ is Donsker and pre-Gaussian uniformly in $P \in \{ P^* : (Y - s(X), D, X) \sim P^*, (Y, D, X) \sim P' \in \mathcal{P}, s \in \{ s' : \mathcal{X} \to \mathcal{Y} \} \}$. 
    This is the case since the covering number of $\mathcal{F}$ for any such $P$ does not depend on $s$. 

    Finally, for verifying pointwise convergence, note that 
    $$\mathbb{E} [ W_i(t) ] = \mathbb{E} [ W_i(t) | \hat{s}_{A,k}] = 0.$$ 
    Hence, it is enough to show that
    \begin{align*}
        \Var \left[ \left( \sqrt{\frac{n_j}{K}} \right)^{-1} \sum_{i \in \mathcal{I}_{j,k}} W_i(t) \right] & = \mathbb{E} \left[\Var \left[\left( \sqrt{\frac{n_j}{K}} \right)^{-1} \sum_{i \in \mathcal{I}_{j,k}} W_i(t) \Bigm| \hat{s}_{A,k} \right] \right] \\
        & = \mathbb{E} \left[\Var \left[ W_i(t) \Bigm| \hat{s}_{A,k} \right] \right] \\
        & = \Var \left[ W(t) \right] \to 0
    \end{align*}
    where the first equality follows from the law of total variance and $\mathbb{E}[W_i(t) \mid \hat{s}_{A,k}] = 0$; the second from the fact that $\{ W_i(t) \}_{i \in \mathcal{I}_{j,k}}$ are i.i.d.\ conditional on $\hat{s}_{A,k}$; and the third from the law of total variance applied to $W(t)$ with $\mathbb{E}[W(t) \mid \hat{s}_{A,k}] = 0$.
    It remains to show $\Var[W(t)] \to 0$. 
    Applying $(a - b)^2 \le 2a^2 + 2b^2$,
    \begin{align*}
        & \Var \left[ W(t) \right] \\
        & \le \mathbb{E} \left[ W(t)^2 \right] \\
        & \le 2 \mathbb{E} \left[ \left| \mathbb{I}(Y_i \le \hat{s}_{A,k}(X_i) + t) - \mathbb{I}(Y_i \le s_{A,P_n}(X_i) + t) \right| \right] \\
        & \hspace{5pt} + 2 \mathbb{E} \left[ \left( F_{P_n,j}(t, \hat{s}_{A,k}) - F_{P_n,j}(t, s_{A,P_n}) \right)^2 \right]. %
    \end{align*}
    
    For any $\eta > 0$, the two indicators in the first term differ only if $|\hat{s}_{A,k}(X_i) - s_{A,P_n}(X_i)| > \eta$ or $Y_i - s_{A,P_n}(X_i) \in (t - \eta, t + \eta]$, so the first term is at most $P_n(|\hat{s}_{A,k}(X) - s_{A,P_n}(X)| > \eta) + \sup_{P \in \mathcal{P}, \, t' \in \mathbb{R}} P(t' - \eta < Y(j) - s_{A,P}(X) \le t' + \eta)$. The second summand vanishes as $\eta \to 0$ by A\ref{as.cont_outcomes}, and for each fixed $\eta$ the first vanishes as $n \to \infty$ by A\ref{as.limit_of_s}. Hence the first term converges to $0$ uniformly in $t$.
    
    The second term is bounded by the first by Jensen's inequality, so it also converges to $0$, and $\Var[W(t)] \to 0$.

    \item \textbf{Step 2: Equation (\ref{eq.th.asymptotic_dist.2})}
    
    The result follows from (see, for example, Lemma 19.24 in \citet{van1998asymptotic}):
    \begin{itemize}
        \item[(i)] Asymptotic equicontinuity of
        $$\mathbb{G}_{A, P, j, k}(t) = \sqrt{\frac{n_j}{K}} \left( \hat{F}_{j,k}(t, s_{A,P_n}) - F_{P_n,j}(t, s_{A,P_n}) \right)$$
        uniformly in $P \in \mathcal{P}$; 
        \item[(ii)] $d_{A,P_n}(\hat{t}_{A}, t_{A,P_n}) \overset{P_n}{\longrightarrow} 0$;
        \item[(iii)] $\int \left[ \left| \mathbb{I}(Y(j) \le s_{A,P_n} + \hat{t}_A) - \mathbb{I}(Y(j) \le s_{A,P_n} + t_{A,P_n}) \right|^2 \right] dP(Y(j)) \overset{P}{\to} 0$, a (marginal) expectation over the distribution of $Y(j)$ taking $\hat{t}_A$ fixed.
    \end{itemize}

    (i) follows from the same argument for stochastic equicontinuity as in Step 1 (see (\ref{eq.th.asymptotic_dist.1})). 
    (ii) follows from consistency of M-estimators (e.g., Corollary 3.2.3 in \citet{vandervaartWeakConvergenceEmpirical2023}) in $d_{A,P_n}$, since $t_{A,P_n}$ is well separated by A\ref{as.uniqueness_t_infty} and $\sup_t | \hat{F}_{j,k}(t, \hat{s}_{A,k}) - F_{P_n,j}(t, s_{A,P_n}) | \overset{P_n}{\longrightarrow} 0$ by Step 1 and the Donsker property of $\mathcal{F}$;
    (iii) follows from (ii), since the integral is at most $d_{A,P_n}(\hat{t}_{A}, t_{A,P_n})$.

    \item \textbf{Step 3: Equation (\ref{eq.th.asymptotic_dist.3})}
    
    We have that
    \begin{align}
        & \sqrt{n} (\hat{\theta}_A - \bar{\theta}_{A,P_n}) \nonumber \\
        & = \sqrt{n} \frac{1}{K} \sum_{k=1}^K \left( \hat{\theta}_{A,k} - \bar{\theta}_{A, P_n, k} \right) \nonumber \\
        & = \sqrt{n} \frac{1}{K} \sum_{k=1}^K \Bigl( \left[ \hat{F}_{1,k}(\hat{t}_{A},\hat{s}_{A,k}) - F_{P_n,1}(\hat{t}_{A},\hat{s}_{A,k}) \right] \nonumber \\
        & \hspace{70pt} - \left[ \hat{F}_{0,k}(\hat{t}_{A},\hat{s}_{A,k}) - F_{P_n,0}(\hat{t}_{A},\hat{s}_{A,k}) \right] \Bigr) \nonumber \\
        & = \sqrt{n} \frac{1}{K} \sum_{k=1}^K \Bigl( \left[ \hat{F}_{1,k}(t_{A,P_n}, s_{A,P_n}) - F_{P_n,1}(t_{A,P_n}, s_{A,P_n}) \right] \nonumber \\
        & \hspace{70pt} - \left[ \hat{F}_{0,k}(t_{A,P_n}, s_{A,P_n}) - F_{P_n,0}(t_{A,P_n}, s_{A,P_n}) \right] \Bigr) \nonumber \\
        & \hspace{15pt} + \sqrt{n} \frac{1}{K} \sum_{k=1}^K \left( \frac{\sqrt{K}}{\sqrt{n_1}} o_{P_n}(1) - \frac{\sqrt{K}}{\sqrt{n_0}} o_{P_n}(1) \right) \label{eq.th.asymptotic_dist.3.1} \\
        & = \frac{1}{\sqrt{\hat{\pi}}} \frac{1}{\sqrt{n_1}} \sum_{i : D_i = 1} \left[ \mathbb{I}(Y_i \le s_{A,P_n}(X_i) + t_{A,P_n}) - F_{P_n,1}(t_{A,P_n}, s_{A,P_n}) \right] \nonumber \\
        & \hspace{15pt} - \frac{1}{\sqrt{1 - \hat{\pi}}} \frac{1}{\sqrt{n_0}} \sum_{i : D_i = 0} \left[ \mathbb{I}(Y_i \le s_{A,P_n}(X_i) + t_{A,P_n}) - F_{P_n,0}(t_{A,P_n}, s_{A,P_n}) \right] \nonumber \\
        & \hspace{15pt} + o_{P_n}(1) \nonumber
    \end{align}
    where (\ref{eq.th.asymptotic_dist.3.1}) follows from (\ref{eq.th.asymptotic_dist.1}) and (\ref{eq.th.asymptotic_dist.2}). 

    Finally, the result follows from a triangular array central limit theorem (e.g., Theorems 27.2 and 27.3 in \citet{billingsley1995probability}), with joint normality following from the Cram\'er--Wold device (e.g., p. 18 in \citet{serfling1980approximation}).
    \end{itemize}
\end{proof}

\begin{proof}[\hypertarget{proof.th.cis}{Proof of Theorem} \ref{th.cis}] \, \\
    The inequalities for the one-sided confidence intervals follow directly from consistency of $(\hat{\sigma}_{L}, \hat{\sigma}_{U})$ and from Theorem \ref{th.asymptotic_dist}, since $\bar{\theta}_{L, P_n} \le \theta_{L,P_n}^* \le \theta_{P_n} \le \theta_{U,P_n}^* \le \bar{\theta}_{U, P_n}$. For the two-sided case, the inequality follows from Theorem \ref{th.asymptotic_dist} and consistency of $(\hat{\sigma}_{L}, \hat{\sigma}_{U}, \hat{\sigma}_{L,U})$ by applying Proposition 3 in \citet{stoye2009more} centering the estimators at the outer bounds $\bar{\theta}_{L, P_n}$ and $\bar{\theta}_{U, P_n}$. Since these are in general not sharp, the equality in \citet{stoye2009more} becomes an inequality.

    Consistency of $(\hat{\sigma}_{L}, \hat{\sigma}_{U}, \hat{\sigma}_{L,U})$ follows from the weak law of large numbers, equicontinuity of $P(Y(j) - s_{A,P_n}(X) \le t)$ by A\ref{as.cont_outcomes}, $|\hat{s}_{A,k}(X) - s_{A,P_n}(X)| \overset{P_n}{\longrightarrow} 0$ by A\ref{as.limit_of_s} and $d_{A,P_n}(\hat{t}_{A}, t_{A,P_n}) \overset{P_n}{\longrightarrow} 0$ (see Step 2 in the proof of Theorem \ref{th.asymptotic_dist}).
\end{proof}

\begin{proof}[\hypertarget{proof.th.exactness}{Proof of Theorem} \ref{th.exactness}] \, \\
    
    Let $A \in \{L,U\}$. $\bar{\theta}_{A,P} - \theta_{A,P}^* \overset{P}{\longrightarrow} 0$ follows since 
    $$\bar{\theta}_{A,P,k} - \theta_{A,P}^* = [F_{P,1}(\hat{t}_{A},\hat{s}_{A,k}) - F_{P,0}(\hat{t}_{A}, \hat{s}_{A,k})] - [F_{P,1}(t_{A,P},s^*_{A,P}) - F_{P,0}(t_{A,P},s^*_{A,P})],$$
    which converges to zero by equicontinuity of $P(Y(j) - s^*_{A,P}(X) \le t)$ by A\ref{as.cont_outcomes}, $|\hat{s}_{A,k}(X) - s^*_{A,P}(X)| \overset{P}{\longrightarrow} 0$ by A\ref{as.limit_of_s} if $s_{A,P} = s^*_{A,P}$, and since $d_{A,P}(\hat{t}_{A}, t_{A,P}) \overset{P}{\longrightarrow} 0$ (see Step 2 in the proof of Theorem \ref{th.asymptotic_dist}). 

    To show $\sqrt{n}(\bar{\theta}_{A,P_n} - \theta_{A,P_n}^*) = o_{P_n}(1)$, note that $F_{P,j}(t, s) = \mathbb{E}[F_j(s(X) + t \mid X)]$, where $F_j(\cdot \mid x)$ is the CDF of $Y(j) \mid X = x$. Hence $\Delta_{P_n}(t, s) = \mathbb{E}[g_X(s(X) + t)]$, where $g_x(u) := F_1(u \mid x) - F_0(u \mid x)$. By construction (see (\ref{def.best_sL})--(\ref{def.best_sU})), $u \mapsto g_x(u)$ is maximized (resp.\ minimized) at $u = s^*_{L,P_n}(x)$ (resp.\ $s^*_{U,P_n}(x)$); under condition (i) of the theorem, $g_x$ is twice continuously differentiable in $u$ with second derivative uniformly bounded.

    Focus on the case of the lower bound (the upper bound is analogous). Define $\hat{\Delta}_L(t) = \frac{1}{K} \sum_{k=1}^K [ \hat{F}_{1,k}(t, \hat{s}_{L,k}) - \hat{F}_{0,k}(t, \hat{s}_{L,k}) ]$, $\bar{\Delta}_{L}(t) = \frac{1}{K} \sum_{k=1}^K \Delta_{P_n}(t, \hat{s}_{L,k})$ and $\hat{G}_L(t) = \sqrt{n} ( \hat{\Delta}_L(t) - \bar{\Delta}_L(t) )$, so that $\hat{t}_L$ maximizes $\hat{\Delta}_L$ and $\bar{\theta}_{L,P_n} = \bar{\Delta}_L(\hat{t}_L)$. Decompose
    \begin{equation*}
        \theta_{L,P_n}^* - \bar{\theta}_{L,P_n} = \bigl[ \theta_{L,P_n}^* - \bar{\Delta}_L(0) \bigr] + \bigl[ \bar{\Delta}_L(0) - \bar{\Delta}_L(\hat{t}_L) \bigr].
    \end{equation*}
    The left-hand side is nonnegative, since each $\Delta_{P_n}(\hat{t}_L, \hat{s}_{L,k}) \le \theta_{L,P_n}^*$ by (\ref{def.induced_bounds}) and Theorem \ref{th.best_s}.

    For the first term, since $\theta_{L,P_n}^* = \Delta_{P_n}(0, s^*_{L,P_n}) = \mathbb{E}[g_X(s^*_{L,P_n}(X))]$,
    \begin{equation*}
        \theta_{L,P_n}^* - \bar{\Delta}_L(0) = \frac{1}{K} \sum_{k=1}^K \mathbb{E}\bigl[ g_X(s^*_{L,P_n}(X)) - g_X(\hat{s}_{L,k}(X)) \bigm| \hat{s}_{L,k} \bigr].
    \end{equation*}
    Since $g_x'(s^*_{L,P_n}(x)) = 0$, a second-order Taylor expansion gives, for some $\tilde{u}(x)$ between $\hat{s}_{L,k}(x)$ and $s^*_{L,P_n}(x)$,
    \begin{equation} \label{eq.proof.th.exactness.taylor}
        0 \le g_x(s^*_{L,P_n}(x)) - g_x(\hat{s}_{L,k}(x)) = \tfrac{1}{2} \bigl| g_x''(\tilde{u}(x)) \bigr| \bigl( \hat{s}_{L,k}(x) - s^*_{L,P_n}(x) \bigr)^2.
    \end{equation}
    By the uniform bound on $g_x''$ and condition (ii), $\sqrt{n} \bigl[ \theta_{L,P_n}^* - \bar{\Delta}_L(0) \bigr] = o_{P_n}(1)$.

    For the second term, since $\hat{\Delta}_L(\hat{t}_L) \ge \hat{\Delta}_L(0)$,
    \begin{equation*}
        \bar{\Delta}_L(0) - \bar{\Delta}_L(\hat{t}_L) = \hat{\Delta}_L(0) - \hat{\Delta}_L(\hat{t}_L) + \frac{\hat{G}_L(\hat{t}_L) - \hat{G}_L(0)}{\sqrt{n}} \le \frac{\hat{G}_L(\hat{t}_L) - \hat{G}_L(0)}{\sqrt{n}}.
    \end{equation*}
    Under the conditions of the theorem, $s_{L,P_n} = s^*_{L,P_n}$, so $0$ maximizes $\Delta_{P_n}(\cdot, s^*_{L,P_n})$; by A\ref{as.uniqueness_t_infty}, $d_{L,P_n}(0, t_{max,P_n}) = 0$, so A\ref{as.uniqueness_t_infty} also holds with $t_{max,P_n}$ replaced by $0$, and I take $t_{L,P_n} = 0$. 
    Hence (\ref{eq.th.asymptotic_dist.1})--(\ref{eq.th.asymptotic_dist.2}) evaluated at $\hat{t}_L$, and (\ref{eq.th.asymptotic_dist.1}) evaluated at $t_{L,P_n} = 0$,
    give for each $j$ and $k$ the same leading term $\sqrt{n_j/K} ( \hat{F}_{j,k}(0, s^*_{L,P_n}) - F_{P_n,j}(0, s^*_{L,P_n}) )$ up to $o_{P_n}(1)$. Since $n / n_j = O_{P_n}(1)$ and $K$ is fixed, $\hat{G}_L(\hat{t}_L) - \hat{G}_L(0) = o_{P_n}(1)$.

    Combining both terms, $0 \le \sqrt{n} \bigl( \theta_{L,P_n}^* - \bar{\theta}_{L,P_n} \bigr) \le o_{P_n}(1)$.
\end{proof}

\end{appendices}

\clearpage
{\singlespacing
\putbib[CAIDE]
}

\end{bibunit}

\clearpage
\begin{bibunit}[ecta]
\begin{center}
{\Large\bfseries Online Appendix}\\
\vspace{0.4em}
{\normalsize Supplemental Appendix to ``Predicting the Distribution of Treatment Effects: A Covariate-Adjustment Approach''}
\end{center}
\vspace{1em}

\begin{appendices}

\setcounter{section}{0}
\renewcommand\appendixname{Online Appendix}
\renewcommand{\thesection}{OA\Alph{section}}
\renewcommand{\thesubsection}{OA\Alph{section}.\arabic{subsection}}

\phantomsection

\subsection{Estimated vs True Propensity Score} \label{appendix.prop_score}

I show that when the propensity score is constant, the asymptotic variances of the estimators in Definition~\ref{def.theta_hat_cf} are smaller than those of alternative estimators that use the true propensity score. 
When the propensity score is known and not constant unlike A\ref{as.unconfoundedness}, an alternative estimator that uses the true propensity score is proposed in Online Appendix \ref{appendix.relaxing_prop_score}(B). 
The lower bound, for example, is given by 
\begin{align*}
    \hat{\theta}^B_{L} & = \max_t \Biggl\{ \frac{1}{n} \sum_{i = 1}^n \biggl( \frac{D_i}{p(X_i)} \mathbb{I}(Y_i \le \hat{s}_{L, k(i)}(X_i) + t) \\
    & \hspace{60pt} - \frac{1 - D_i}{1 - p(X_i)} \mathbb{I}(Y_i \le \hat{s}_{L, k(i)}(X_i) + t) \biggr) \Biggr\}.
\end{align*}

Note that the estimators in Definition \ref{def.theta_hat_cf} can be rewritten in terms of an \textit{estimated} propensity score $\hat{p} = \frac{1}{n} \sum_{i=1}^n D_i$. 
For the lower bound, we have  
\begin{align*}
    \hat{\theta}_{L} & \equiv \max_t \Biggl\{ \frac{1}{n} \sum_{i = 1}^n \biggl( \frac{D_i}{\hat{p}} \mathbb{I}(Y_i \le \hat{s}_{L, k(i)}(X_i) + t) \\
    & \hspace{60pt} - \frac{1 - D_i}{1 - \hat{p}} \mathbb{I}(Y_i \le \hat{s}_{L, k(i)}(X_i) + t) \biggr) \Biggr\}.
\end{align*}

Let $t_{L,P} = t_{max,P}$; 
$t_{U,P} = t_{min,P}$; 
$\pi = P(D = 1)$ (as in A\ref{as.rct}). 
For $A \in \{L, U\}$ and $j\in \{0, 1 \}$, define $Z_{A, P}(j) = \mathbb{I}(Y(j) \le s_{A,P}(X) + t_{A,P})$. 
Then, the asymptotic variance of $\hat{\theta}_{A}$ (as in Definition \ref{def.theta_hat_cf}) is given by:
$$\frac{\Var[Z_{A,P}(1)]}{\pi} + \frac{\Var[Z_{A,P}(0)]}{1 - \pi}$$

If $p(x) = \pi$, the asymptotic variance of $\hat{\theta}^B_{A}$ simplifies to:
$$\frac{\Var[Z_{A,P}(1)]}{\pi} + \frac{\Var[Z_{A,P}(0)]}{1 - \pi} + \left( \frac{\mathbb{E}[Z_{A,P}(1)]}{\pi} + \frac{\mathbb{E}[Z_{A,P}(0)]}{1 - \pi} \right)^2 \pi (1 - \pi)$$

Since the third term above is never negative and, in general, different from zero, the asymptotic variance of $\hat{\theta}_{A}$ is smaller than that of $\hat{\theta}^B_{A}$.

\subsection{Asymptotic Distribution of Plug-in Estimator Based on \texorpdfstring{\citet{fan2010sharp}}{Fan and Park (2010)}} \label{appendix.fanpark}

I argue that it is generally difficult to characterize the asymptotic distribution of a plug-in estimator based on \citet{fan2010sharp}'s representation of the identified set (\ref{def.theta_sharp}) when the conditional distributions are estimated nonparametrically. 
I show that the asymptotic distribution depends on the rate of convergence of the estimators of $F_1(t | x)$ and $F_0(t | x)$, and that different models/algorithms used to estimate these functions lead to different asymptotic distributions of the plug-in estimator. 
This is true even when using sample-splitting. 
For simplicity, I focus on the lower bound. 

Define 
$$\Delta(t | x) = F_1(t | x) - F_0(t | x).$$
Let $\hat{F}_1(t | x)$ and $\hat{F}_0(t | x)$ be estimators respectively of $F_1(t | x)$ and $F_0(t | x)$, and $\hat{\Delta}(t | x) = \hat{F}_1(t | x) - \hat{F}_0(t | x)$. 
Define $\hat{s}(x) \in \arg\max_{t} \hat{\Delta}(t | x)$. 
Since my estimation strategies in Sections \ref{section.finite_sample} and \ref{section.asymptotics} use sample-splitting, I also use sample-splitting to define the plug-in estimator based on (\ref{def.theta_sharp}). 
Consider a random split of $\{ 1, \dots, n\}$ into three equal-sized sets $\mathcal{I}_1$, $\mathcal{I}_2$, and $\mathcal{I}_3$ ($n$ is assumed a multiple of 3 for simplicity). 
Let $\hat{\Delta}$ be estimated using data from $\mathcal{I}_1$ only, and $\hat{s}$ estimated using data from $\mathcal{I}_2$ only. 
The plug-in estimator of $\theta_L^*$ based on (\ref{def.theta_sharp}) is given by 
$$\hat{\theta}_{FP} = \frac{1}{n / 3} \sum_{i \in \mathcal{I}_3} \hat{\Delta}(\hat{s}(X_i) | X_i).$$

Let $\mathcal{D}_j = (Y_i, D_i, X_i)_{i \in \mathcal{I}_j}$ for $j=1,2,3$. 
Define 
$$\bar{\Delta}(\mathcal{D}_1, \mathcal{D}_2) = \mathbb{E} [ \hat{\Delta}(\hat{s}(X) | X) | \mathcal{D}_1, \mathcal{D}_2],$$
where $X$ is independent of $\mathcal{D}_1, \mathcal{D}_2$. 
It follows that 
$$\left( \hat{\theta}_{FP} - \theta_L^* \right) = \left( \frac{1}{n / 3} \sum_{i \in \mathcal{I}_3} \hat{\Delta}(\hat{s}(X_i) | X_i) - \bar{\Delta}(\mathcal{D}_1, \mathcal{D}_2) \right) + \left( \bar{\Delta}(\mathcal{D}_1, \mathcal{D}_2) - \theta_L^* \right).$$
Since the first term is mean-zero, it follows that 
$$\sqrt{n} \left( \frac{1}{n / 3} \sum_{i \in \mathcal{I}_3} \hat{\Delta}(\hat{s}(X_i) | X_i) - \bar{\Delta}(\mathcal{D}_1, \mathcal{D}_2) \right)$$
is asymptotically normal by a triangular array CLT. 
\sloppy
The second term, $\bar{\Delta}(\mathcal{D}_1, \mathcal{D}_2) - \theta_L^*$, is random since the data $(\mathcal{D}_1, \mathcal{D}_2)$ is random. 
$\bar{\Delta}(\mathcal{D}_1, \mathcal{D}_2) - \theta_L^*$ may or may not be bounded in probability when multiplied by $\sqrt{n}$. 
In general, the rate of convergence of this term depends on the rate of convergence of $\hat{\Delta}$. 
Since $\hat{\Delta}$ is estimated nonparametrically, this rate may be slow when regularity conditions such as sparsity and specific forms of smoothness do not hold. 
In contrast, the asymptotic distribution in Theorem \ref{th.asymptotic_dist_simplified} does not rely on such conditions. 
The asymptotic distribution of an appropriately scaled version of $\bar{\Delta}(\mathcal{D}_1, \mathcal{D}_2) - \theta_L^*$ depends on which model/algorithm is used to estimate $\hat{\Delta}$, since different estimators $\hat{\Delta}$ lead to different values of $\bar{\Delta}$. 
For example, the distribution of $\bar{\Delta}(\mathcal{D}_1, \mathcal{D}_2)$ is different and difficult to characterize if $\hat{\Delta}$ is calculated with random forests or neural networks, which are some of the models I use in Section \ref{section.application}. 

Note that the term $\bar{\Delta}(\mathcal{D}_1, \mathcal{D}_2) - \theta_L^*$ may be positive or negative. 
If $\bar{\Delta}(\mathcal{D}_1, \mathcal{D}_2)$ is symmetrically distributed around $\theta_L^*$, $\bar{\Delta}(\mathcal{D}_1, \mathcal{D}_2) - \theta_L^* > 0$ with probability $1/2$. 
In this case, a hypothesis test for $H_0: \theta_L^* \le \tau$, $\tau \in [0,1]$, based on the asymptotic distribution of $\sqrt{n} \left( \frac{1}{n / 3} \sum_{i \in \mathcal{I}_3} \hat{\Delta}(\hat{s}(X_i) | X_i) - \bar{\Delta}(\mathcal{D}_1, \mathcal{D}_2) \right)$ may overreject. 
That is, if one conducts a hypothesis test with significance level $\alpha$ based on the incorrect assumption that $\sqrt{n} \left( \bar{\Delta}(\mathcal{D}_1, \mathcal{D}_2) - \theta_L^* \right) = o_P(1)$, the test may, for example, reject $\theta_L^* = 0$ when $\theta = 0$ with probability higher than $\alpha$. 

In contrast, my approaches to inference do not depend on the rate of convergence or the specific model/algorithm used to estimate $s^*_L$. 
Theorem \ref{th.fs_cis} provides a finite-sample guarantee of coverage under no restrictions on how $\hat{s}_L$ is calculated. 
In Theorem \ref{th.cis}, I provide a large-sample guarantee of coverage under no restrictions on the rate of convergence of $\hat{s}_L$ or which model/algorithm is used to calculate it. 
The result of Theorem \ref{th.cis} is possible because of the alternative characterization of the identified set in Theorem \ref{th.best_s}. 
Unlike (\ref{def.theta_sharp}), which relies on the average of conditional expectations, (\ref{def.induced_bounds}) is expressed in terms of unconditional expectations. 
As a consequence, Theorem \ref{th.asymptotic_dist_simplified} shows that estimation error in $\hat{s}_L$ translates into negative bias, 
$$\mathbb{E}[\sqrt{n} (\hat{\theta}_L - \theta_L^*)] \le 0 + o(1),$$
which makes the estimate $\hat{\theta}_L$ conservative. 
Since $\theta_L^*$ is a lower bound, an estimate $\hat{\theta}_L$ that is too low leads to a confidence set that may be conservative, but valid. 

\subsection{Estimation of the Covariate-Adjustment Term} \label{appendix.ml}

I consider two approaches to estimate the covariate-adjustment term $s(x)$. 
For simplicity, I focus the explanation on the case of the lower bound. 
First, I estimate $\hat{F}_j(t|x)$ ($j=0,1$) with one of two methods, and then I calculate $\hat{s}$ from $\hat{F}_j(t|x)$. 

The first method for estimating $\hat{F}_j(t|x)$ is using quantile models, such as quantile forests or neural networks. 
They give estimates for quantiles, that is, $\hat{F}^{-1}_j(\tau|x)$, for all values of $x$ in the sample. 
$\hat{F}_j(t|x)$ can be calculated from $\hat{F}^{-1}_j(\tau|x)$ by linear interpolation. 
In Sections \ref{section.simulation} and \ref{section.application}, I consider a grid for $\tau$ from $0$ to $1$ by increments of $0.01$, i.e., $\tau_k \in \{0, 0.01, \dots, 1 \}$.
For an arbitrary $t$, pick $\tau_a$ and $\tau_b$ such that $t_a = \hat{F}^{-1}_j(\tau_a|x), t_b = \hat{F}^{-1}_j(\tau_b|x)$ are the two closest values to $t$ such that $t_a \le t \le t_b$. 
Then, $\hat{F}_j(t|x) = \tau_a + \frac{\tau_b - \tau_a}{t_b - t_a} (t - t_a)$.

The second method for estimating $\hat{F}_j(t|x)$ assumes that the distribution of $Y(j) | X = x$ is the same for all $x$ up to a shift in the mean $\mathbb{E}[Y(j)|X=x]$. This is allowed by the framework of Theorems \ref{th.fs_cis} and \ref{th.asymptotic_dist_simplified} since misspecification of $s^*$ leads to wider but still valid bounds. Denote by $\hat{\mu}(x)$ an estimator of $\mathbb{E}[Y(j)|X=x]$, using any regression method. For example, in Section \ref{section.simulation} I use a support vector machine. Denote by $\hat{F}_e(t)$ the empirical CDF of the residuals $Y(j) - \hat{\mu}(X)$ inside the training sample. Then, $\hat{F}_j(t|x) = \hat{F}_e(t - \hat{\mu}(x))$.

Once $\hat{F}_j(t|x)$ is estimated, the covariate-adjustment term is calculated by $\hat{s}(x) = \arg \max_{t \in \mathcal{T}} \left\{ \hat{F}_1(t | x) - \hat{F}_0(t | x) \right\}$. For $\mathcal{T}$, I use a random grid of size $10^4$ drawn from a normal distribution with mean $0$ and standard deviation $\max(Y) - \min(Y)$.

I fit all models in \textsf{R}. For quantile neural networks, I use the package \texttt{qrnn} with 3 hidden nodes, in 1 trial and a maximum of 100 iterations. 
Other parameters are the default choice. 
For quantile random forests, I use the package \texttt{ranger} with 1,000 trees and otherwise default hyperparameters. For all the regression methods, I use the package \texttt{mlr3} with default values of hyperparameters. The packages for all models, used through \texttt{mlr3}, are: random forests (\texttt{ranger}), extreme gradient boosting (\texttt{xgboost}), support vector machine (\texttt{e1071}), neural networks (\texttt{nnet}), elastic net (\texttt{glmnet}). The constant/none model takes $s(x) = 0$.

For choosing the best model in the application of Section \ref{section.application}, I use cross-validation inside each fold. 
That is, for each fold, I fit all the models in a 10-fold cross-validation and choose the model that gives the largest (smallest) lower (upper) bound. 
Note that the test set is never used in this approach.

When calculating the sample-splitting estimators and confidence intervals of Section \ref{section.finite_sample}, I use a 50-50 split of the sample. 

\subsection{Relaxing the Constant Propensity Score} \label{appendix.relaxing_prop_score}

I propose two alternatives to the cross-fitting estimators in Definition \ref{def.theta_hat_cf}. 
For simplicity, I focus on the lower bound. 

\begin{itemize}
    \item[(A)] \textbf{Constant propensity score within groups:} Let there be a fixed number $\bar{G} \ge 1$ of groups, such that the propensity score is constant within each group (and bounded away from zero and one). 
    Define $n_{1, g}$ as the number of observations in group $g$ assigned to treatment and $n_{0, g}$ the number of observations assigned to control.  
    Finally, let $G_i \in \{1, \dots, \bar{G} \}$ be the group to which observation $i$ belongs, where group membership is a known function of the covariates $X_i$, so group sizes are random. 
    Then, a new cross-fitting estimator for the lower bound can be defined as 
    \begin{align*}
        \widehat{\theta}^A_{L} & = \max_t \Biggl\{ \sum_{g = 1}^{\bar{G}} \frac{n_{1,g} + n_{0,g}}{n} \Biggl[ \frac{1}{n_{1,g}} \sum_{\substack{i : D_i = 1 \\ G_i = g}} \mathbb{I}(Y_i \le \hat{s}_{L, k(i)}(X_i) + t) \\
        & \hspace{60pt} - \frac{1}{n_{0,g}} \sum_{\substack{i : D_i = 0 \\ G_i = g}} \mathbb{I}(Y_i \le \hat{s}_{L, k(i)}(X_i) + t) \Biggr] \Biggr\},
    \end{align*}
    where the data is split such that for each fold $k$, there is balance in the number of treated and control units within each group. 
    The asymptotic distribution follows very similarly to the proof of Theorem \ref{th.asymptotic_dist_simplified} by applying the arguments in (\ref{eq.th.asymptotic_dist.1}) and (\ref{eq.th.asymptotic_dist.2}) for each group-fold combination, instead of for each fold. 
    The asymptotic variance-covariance matrix of the estimators is given by
    \begin{align*}
        \sigma^2_{A,P} = & \sum_{g=1}^{\bar{G}} \pi(g) \left( \frac{1}{\pi_1(g)} \Var[Z_{A}(1) | G = g] + \frac{1}{\pi_0(g)} \Var[Z_{A}(0) | G = g] \right) \\
        & + \sum_{g=1}^{\bar{G}} \pi(g) \left( \Delta_A(g) - \theta_{A,P} \right)^2,
    \end{align*}
    \begin{align*}
        \sigma_{L,U,P} = & \sum_{g=1}^{\bar{G}} \pi(g) \left( \frac{1}{\pi_1(g)} \Cov[Z_{L}(1), Z_{U}(1) | G = g] + \frac{1}{\pi_0(g)} \Cov[Z_{L}(0), Z_{U}(0) | G = g] \right) \\
        & + \sum_{g=1}^{\bar{G}} \pi(g) \left( \Delta_L(g) - \theta_{L,P} \right) \left( \Delta_U(g) - \theta_{U,P} \right),
    \end{align*}
    where $A \in \{L, U\}$, and 
    $$\pi(g) = P(G = g), \quad \pi_j(g) = P(D = j | G = g),$$
    $$Z_{A}(j) = \mathbb{I}(Y(j) \le s_{A,P}(X) + t_{A,P}), \quad \Delta_A(g) = \mathbb{E}[Z_{A}(1) - Z_{A}(0) | G = g].$$
    
    Note that, as before, the elements of the variance-covariance matrix can be estimated using sample analogues.
    \item[(B)] \textbf{Known propensity score:} Assume the propensity score $p(x) = P(D = 1 | X = x)$ is known and bounded away from zero and one. 
    Then, the estimator for the lower bound can be defined as
    \begin{align*}
        \widehat{\theta}^B_{L} & = \max_t \Biggl\{ \frac{1}{n} \sum_{i = 1}^n \biggl( \frac{D_i}{p(X_i)} \mathbb{I}(Y_i \le \hat{s}_{L, k(i)}(X_i) + t) \\
        & \hspace{60pt} - \frac{1 - D_i}{1 - p(X_i)} \mathbb{I}(Y_i \le \hat{s}_{L, k(i)}(X_i) + t) \biggr) \Biggr\}.
    \end{align*}

    The proof of Theorem \ref{th.asymptotic_dist_simplified} extends to this case with minor adjustments. 
    Note that since $p(x)$ is (uniformly) bounded away from zero and one, the same Donsker properties hold for $(\widehat{\theta}^B_{L}, \widehat{\theta}^B_{U})$. 
    The asymptotic variance-covariance matrix of the estimators is given by
    $$\sigma^2_{A,P} = \mathbb{E}_P \left[ \left( \frac{D}{p(X)^2} + \frac{1 - D}{(1 - p(X))^2} \right) \mathbb{I}(Y \le s_{A,P}(X) + t_{A,P}) \right] - \theta_{A,P}^2,$$
    \begin{align*}
        & \sigma_{L,U,P} = \\
        & \mathbb{E}_P \left[ \left( \frac{D}{p(X)^2} + \frac{1 - D}{(1 - p(X))^2} \right) \mathbb{I}(Y \le s_{L,P}(X) + t_{L,P}) \mathbb{I}(Y \le s_{U,P}(X) + t_{U,P}) \right] \\
        & - \theta_{L,P} \theta_{U,P},
    \end{align*}
    where $A \in \{L, U\}$ and $p(x) = P(D = 1 | X = x)$.
\end{itemize}

\clearpage

\subsection{Definition of Two-Sided Confidence Intervals} \label{appendix.stoye}

The definition of $\widehat{CI}_\alpha$ follows from \citet{stoye2009more}. 
Let $h_n$ denote a preassigned sequence such that $h_n \to 0$ and $\sqrt{n} h_n \to \infty$, and define 
$$\Lambda = \begin{cases} \widehat{\theta}_U - \widehat{\theta}_L, & \widehat{\theta}_U - \widehat{\theta}_L > h_n \\ 0, & \text{otherwise} \end{cases}$$
Then, the CI is defined by
$$\widehat{CI}_\alpha = \begin{cases} \left[ \widehat{\theta}_L - \widehat{c}_L \frac{\widehat{\sigma}_L}{\sqrt{n}}, \widehat{\theta}_U + \widehat{c}_U \frac{\widehat{\sigma}_U}{\sqrt{n}} \right], & \text{ if } \widehat{\theta}_L - \widehat{c}_L \frac{\widehat{\sigma}_L}{\sqrt{n}} \le \widehat{\theta}_U + \widehat{c}_U \frac{\widehat{\sigma}_U}{\sqrt{n}} \\ \emptyset, & \text{ otherwise} \end{cases}$$
where $(\widehat{c}_L, \widehat{c}_U)$ minimize $(\widehat{\sigma}_L \widehat{c}_L + \widehat{\sigma}_U \widehat{c}_U)$ subject to the constraints that
\begin{align*}
    P \left( -\widehat{c}_L \le Z_1 \land \frac{\widehat{\sigma}_{L,U}}{\widehat{\sigma}_{L} \widehat{\sigma}_{U}} Z_1 \le \widehat{c}_U + \frac{\sqrt{n} \Lambda}{\widehat{\sigma}_{U}} + Z_2 \sqrt{1 - \frac{\widehat{\sigma}^2_{L,U}}{\widehat{\sigma}^2_{L} \widehat{\sigma}^2_{U}}} \right) & \ge 1 - \alpha \\
    P \left( -\widehat{c}_L - \frac{\sqrt{n} \Lambda}{\widehat{\sigma}_{L}} + Z_2 \sqrt{1 - \frac{\widehat{\sigma}^2_{L,U}}{\widehat{\sigma}^2_{L} \widehat{\sigma}^2_{U}}} \le \frac{\widehat{\sigma}_{L,U}}{\widehat{\sigma}_{L} \widehat{\sigma}_{U}} Z_1 \land Z_1 \le \widehat{c}_U \right) & \ge 1 - \alpha
\end{align*}
and $Z_1$, $Z_2$ are independent standard normal random variables.

Importantly, calculating $\widehat{CI}_\alpha$ relies on choosing the hyperparameter $h_n$. A possible choice is Stoye's suggestion of $h_n = n^{-1/2} (\log \log n)^{1/2}$, and other possibilities are $h_n = n^{-1/2} (\log n)^{1/2}$ and $h_n = n^{-1/2} (Q \log \log n)^{1/2}$ for $Q \ge 2$, as suggested in \citet{andrews2010inference}. 
Discussing the best choice for $h_n$ is beyond the scope of this paper, but a practical procedure for applications is to assess the sensitivity of the CI to different choices of $h_n$.

\clearpage

\subsection{Relaxing Assumption \ref{as.uniqueness_t_infty}} \label{appendix.relaxing_maximizer}

I propose two alternatives to Assumption \ref{as.uniqueness_t_infty}. 
The first requires changing the definition of the estimators $\widehat{\theta}_L$ and $\widehat{\theta}_U$, while the second keeps the original estimators and modifies the critical values. 

\begin{itemize}
    \item[(A)] \textbf{Changing the estimators:}
    The first alternative is to change the definition of $\widehat{\theta}_L$ and $\widehat{\theta}_U$ to avoid the optimization problem. 
    For example, $t$ could be fixed at $t=0$ or learned from other folds inside the cross-fitting procedure. 
    Focusing on the case of $\widehat{\theta}_L$, when estimating $\hat{s}_{L, k}$, one could estimate $\hat{t}_k$ (with data from all except the $k$-th fold) which maximizes the expression inside the $k$-th fold, and estimate the lower bound as 
    $$\frac{1}{n_1} \sum_{i : D_i = 1} \mathbb{I}(Y_i \le \hat{s}_{L, k(i)}(X_i) + \hat{t}_{k(i)}) - \frac{1}{n_0} \sum_{i : D_i = 0} \mathbb{I}(Y_i \le \hat{s}_{L, k(i)}(X_i) + \hat{t}_{k(i)})$$

    This would be valid since it is equivalent to a new estimator $\tilde{s}_{L, k} = \hat{s}_{L, k} + \hat{t}_k$, and there is no optimization in the final estimator. 
    
    \item[(B)] \textbf{Changing the critical values:}
    The second alternative keeps $\hat{\theta}_L$ and $\hat{\theta}_U$ and changes how critical values are computed. 
    I describe three procedures: two based on the intersection-bounds approach of \citet{chernozhukovIntersectionBoundsEstimation2013} (CLR), and one based on the bootstrap of directional derivatives of \citet{fangInferenceDirectionallyDifferentiable2019} (FS). 
    I focus on the lower bound, and the upper bound is analogous. 

    As in the proof of Theorem \ref{th.exactness}, let $\hat{\Delta}_L(t) = \hat{F}_{L,1}(t) - \hat{F}_{L,0}(t)$, so that $\hat{\theta}_L = \sup_t \hat{\Delta}_L(t)$, and $\bar{\Delta}_L(t) = \frac{1}{K} \sum_{k=1}^K \Delta_P(t, \hat{s}_{L,k})$. 
    All three procedures take critical values from the multiplier process 
    $$\hat{G}^*_L(t) = \sqrt{n} \sum_{j=0}^1 \frac{2j - 1}{n_j} \sum_{i : D_i = j} \xi_i \left[ \mathbb{I}(Y_i \le \hat{s}_{L,k(i)}(X_i) + t) - \hat{F}_{L,j}(t) \right],$$
    where $\xi_i \sim N(0,1)$ are i.i.d.\ and independent of the data, and $Q^*_{1-\alpha}$ denotes the $1 - \alpha$ quantile conditional on the data.

    \begin{itemize}
        \item[(i)] \textit{CLR with all thresholds.} 
        Let $\hat{\sigma}_L(t)$ be the estimator $\hat{\sigma}_L$ of Appendix \ref{appendix.definitions} with $\hat{t}_L$ replaced by $t$, truncated below at a fixed $\underline{\sigma} > 0$. 
        The precision-corrected estimator of \citet[eq.~(3.7)]{chernozhukovIntersectionBoundsEstimation2013}, applied to $-\hat{\Delta}_L$ and using all thresholds, gives 
        $$L_\alpha^{\mathrm{all}} = \sup_t \left\{ \hat{\Delta}_L(t) - \hat{c}^{\mathrm{all}} \frac{\hat{\sigma}_L(t)}{\sqrt{n}} \right\}, \quad \hat{c}^{\mathrm{all}} = Q^*_{1-\alpha} \left( \sup_t \frac{\hat{G}^*_L(t)}{\hat{\sigma}_L(t)} \right).$$

        \item[(ii)] \textit{Adaptive inequality selector CLR.} 
        The critical value in (i) accounts for thresholds far from the maximizers, which can make $L_\alpha^{\mathrm{all}}$ conservative. 
        Following the adaptive inequality selection of \citet[eq.~(3.5)]{chernozhukovIntersectionBoundsEstimation2013}, let $\hat{M}_L = \{ t : \hat{\theta}_L - \hat{\Delta}_L(t) \le h_n \}$ for a tuning parameter $h_n > 0$ with $\sqrt{n} h_n \to \infty$ and $\sup_t | \bar{\Delta}_L(t) - \Delta_P(t, s_{L,P}) | = o_P(h_n)$, and define $L_\alpha^{\mathrm{ad}}$ as $L_\alpha^{\mathrm{all}}$ with $\hat{c}^{\mathrm{all}}$ replaced by 
        $$\hat{c}^{\mathrm{ad}} = Q^*_{1-\alpha} \left( \sup_{t \in \hat{M}_L} \frac{\hat{G}^*_L(t)}{\hat{\sigma}_L(t)} \right).$$

        \item[(iii)] \textit{FS.} 
        The map $h \mapsto \sup_t h(t)$ is Hadamard directionally differentiable at $\Delta_P(\cdot, s_{L,P})$, with derivative $h \mapsto \sup_{t \in M_{L,P}} h(t)$, where $M_{L,P}$ is the set of maximizers of $\Delta_P(\cdot, s_{L,P})$ \citep[Theorem 2.1 and Corollary 2.8]{carcamo2020directional}.
        This derivative is nonlinear when $M_{L,P}$ contains thresholds with different scores, in which case the standard bootstrap of $\hat{\theta}_L$ is inconsistent \citep[Theorem 3.1]{fangInferenceDirectionallyDifferentiable2019}. 
        Following \citet[eq.~(22)]{fangInferenceDirectionallyDifferentiable2019}, an alternative is to estimate the derivative by $h \mapsto \sup_{t \in \hat{M}_L} h(t)$, with $\hat{M}_L$ as in (ii), and let 
        $$L_\alpha^{\mathrm{FS}} = \hat{\theta}_L - \frac{\hat{q}^{\mathrm{FS}}}{\sqrt{n}}, \quad \hat{q}^{\mathrm{FS}} = Q^*_{1-\alpha} \left( \sup_{t \in \hat{M}_L} \hat{G}^*_L(t) \right).$$
    \end{itemize}

    Both CLR and FS require that, for some tight Gaussian process $\mathbb{G}_P$ and uniformly in $t$,
    $$\hat{G}_L(t) = \sqrt{n} \bigl( \hat{\Delta}_L(t) - \bar{\Delta}_L(t) \bigr) \overset{d}{\to} \mathbb{G}_P(t), \quad \text{and} \quad \hat{G}^*_L(t) \overset{d}{\to} \mathbb{G}_P(t)$$
    conditional on the data.
    Step 1 in the proof of Theorem \ref{th.asymptotic_dist}, which does not use A\ref{as.uniqueness_t_infty}, shows that $\hat{G}_L$ is asymptotically equivalent to an empirical process indexed by the Donsker class of indicators $\mathcal{F}$, which gives the first convergence. 
    The second follows from the conditional multiplier central limit theorem \citep[Theorem 2.9.6]{vandervaartWeakConvergenceEmpirical2023}. 
    Two features differ from their settings. 
    First, the centering $\bar{\Delta}_L$ depends on the first-stage estimates, so the procedures give lower confidence limits for $\sup_t \bar{\Delta}_L(t)$, which is at most $\theta$ by (\ref{def.induced_bounds}). 
    Second, $\bar{\Delta}_L$ converges to $\Delta_P(\cdot, s_{L,P})$ only at the rate of $\hat{s}_{L,k}$, which is why the condition on $h_n$ in (ii) is stronger than the threshold of order $n^{-1/2}$ (up to a slowly growing factor) in \citet[eq.~(3.5)]{chernozhukovIntersectionBoundsEstimation2013}; it ensures that $\hat{M}_L$ contains $M_{L,P}$ with probability approaching one. 
    Procedures (ii) and (iii) also require the maximizers to be well separated, that is, A\ref{as.uniqueness_t_infty} with $t_{max,P}$ replaced by the set $M_{L,P}$. 
\end{itemize}

\clearpage

\subsection{Additional Tables} \label{appendix.additional_tables}

\begin{table}[ht]
    \centering
    \caption{Bounds on fraction harmed by microcredit - alternative values of $\delta$}
    \label{tab.deltas}
    \begin{threeparttable}
    \begin{tabular}{@{}l|c|c@{}}
    \toprule
    & \begin{tabular}{@{}c@{}}Lower Bound\\($\delta = -0.05$)\end{tabular} & \begin{tabular}{@{}c@{}}Upper Bound\\($\delta = 0.05$)\end{tabular} \\
    \midrule
    \midrule
    Bosnia & 
    \begin{tabular}[c]{@{}c@{}}0.053\\(0.027)\\{[}0.027{]}\end{tabular} & 
    \begin{tabular}[c]{@{}c@{}}0.971\\(0.018)\\{[}0.054{]}\end{tabular} \\ \midrule
    Egypt & 
    \begin{tabular}[c]{@{}c@{}}0.182\\(0.022)\\{[}0.000{]}\end{tabular} & 
    \begin{tabular}[c]{@{}c@{}}0.749\\(0.022)\\{[}0.000{]}\end{tabular} \\ \midrule
    Philippines (1) & 
    \begin{tabular}[c]{@{}c@{}}0.034\\(0.038)\\{[}0.181{]}\end{tabular} & 
    \begin{tabular}[c]{@{}c@{}}0.950\\(0.039)\\{[}0.101{]}\end{tabular} \\ \midrule
    Philippines (2) & 
    \begin{tabular}[c]{@{}c@{}}0.082\\(0.032)\\{[}0.006{]}\end{tabular} & 
    \begin{tabular}[c]{@{}c@{}}0.853\\(0.009)\\{[}0.000{]}\end{tabular} \\ \midrule
    South Africa & 
    \begin{tabular}[c]{@{}c@{}}0.022\\(0.020)\\{[}0.130{]}\end{tabular} & 
    \begin{tabular}[c]{@{}c@{}}0.873\\(0.033)\\{[}0.000{]}\end{tabular} \\
    \bottomrule
    \end{tabular}
    \begin{tablenotes}
    \small
    \item Standard errors in parentheses, and p-values in brackets. The null hypothesis for the lower bound is $\theta^*_L = 0$, and for the upper bound $\theta^*_U = 1$. Response variable: $\log(1 + household\ income )$.
    \end{tablenotes}
    \end{threeparttable}
\end{table}

\begin{table}[htbp]
    \begin{threeparttable}[ht]
    \centering
    \caption{Bounds on fraction harmed by microcredit - business profits}
    \label{tab.bprofits}
    \begin{tabular}{@{}l>{\begin{tabular}[c]{@{}c@{}}\end{tabular}}c>{\begin{tabular}[c]{@{}c@{}}\end{tabular}}c>{\begin{tabular}[c]{@{}c@{}}\end{tabular}}c>{\begin{tabular}[c]{@{}c@{}}\end{tabular}}c@{}}
    \toprule
    & \begin{tabular}[c]{@{}l@{}}Lower Bound\end{tabular} & \begin{tabular}[c]{@{}c@{}}Lower Bound\\(No Covariates)\end{tabular} & \begin{tabular}[c]{@{}c@{}}Upper Bound\end{tabular} & \begin{tabular}[c]{@{}c@{}}Upper Bound\\(No Covariates)\end{tabular} \\ \midrule
    Bosnia & 
    \begin{tabular}[c]{@{}c@{}}0.012\\ {[}0.354{]}\end{tabular} & 
    \begin{tabular}[c]{@{}c@{}}0.006\\ {[}0.171{]}\end{tabular} & 
    \begin{tabular}[c]{@{}c@{}}0.918\\ {[}0.003{]}\end{tabular} & 
    \begin{tabular}[c]{@{}c@{}}0.936\\ {[}0.007{]}\end{tabular} \\ \midrule
    Egypt & 
    \begin{tabular}[c]{@{}c@{}}0.249\\ {[}0.000{]}\end{tabular} & 
    \begin{tabular}[c]{@{}c@{}}0.008\\ {[}0.328{]}\end{tabular} & 
    \begin{tabular}[c]{@{}c@{}}0.828\\ {[}0.000{]}\end{tabular} & 
    \begin{tabular}[c]{@{}c@{}}0.958\\ {[}0.008{]}\end{tabular} \\ \midrule
    Philippines (1) & 
    \begin{tabular}[c]{@{}c@{}}0.077\\ {[}0.025{]}\end{tabular} & 
    \begin{tabular}[c]{@{}c@{}}0.071\\ {[}0.036{]}\end{tabular} & 
    \begin{tabular}[c]{@{}c@{}}0.902\\ {[}0.006{]}\end{tabular} & 
    \begin{tabular}[c]{@{}c@{}}0.931\\ {[}0.005{]}\end{tabular} \\ \midrule
    Philippines (2) & 
    \begin{tabular}[c]{@{}c@{}}0.094\\ {[}0.002{]}\end{tabular} & 
    \begin{tabular}[c]{@{}c@{}}0.046\\ {[}0.085{]}\end{tabular} & 
    \begin{tabular}[c]{@{}c@{}}0.961\\ {[}0.005{]}\end{tabular} & 
    \begin{tabular}[c]{@{}c@{}}0.981\\ {[}0.089{]}\end{tabular} \\
    \bottomrule
    \end{tabular}
    \begin{tablenotes}
    \small
    \item p-values in brackets. The null hypothesis for the lower bound is $\theta^*_L = 0$, and for the upper bound $\theta^*_U = 1$. $1 -$ Upper Bound represents fraction benefited by microcredit. South Africa is excluded as business profits were not measured in that study.
    \end{tablenotes}
    \end{threeparttable}
\end{table}

\end{appendices}

\clearpage
\renewcommand{\baselinestretch}{1}
\putbib[CAIDE]

\end{bibunit}
\end{document}